\documentclass[12pt]{article}
\usepackage{amssymb,amsmath,amsthm,epsfig,bm}
\numberwithin{equation}{section}
\numberwithin{figure}{section}
\usepackage{epic,graphicx,color}
\usepackage{ifpdf}
\usepackage{mathrsfs}
\usepackage{chemformula}

\usepackage[noend]{algpseudocode}
\usepackage{algorithmicx,algorithm}

\usepackage{multirow}
\usepackage{makecell}
\usepackage{booktabs}
\usepackage{dsfont}
\usepackage{tabularx}
\usepackage{diagbox}
\usepackage{caption,subcaption}
\newcommand{\tabincell}[2]{\begin{tabular}{@{}#1@{}}#2\end{tabular}}
\usepackage{threeparttable}
\usepackage[numbers,sort&compress]{natbib}

\theoremstyle{plain}
\newtheorem{lem}{Lemma}[section]
\newtheorem{thm}[lem]{Theorem}

\theoremstyle{definition}

\theoremstyle{remark}

\usepackage{indentfirst}

\begin{document}
	
	\title{ \large\bf Conditional variational autoencoder with Gaussian process regression recognition for parametric models}
	
	\author{Xuehan Zhang\thanks{School of Mathematical Sciences,  Tongji University, Shanghai 200092, China. ({\tt  xhzhang@tongji.edu.cn}).}
		\and
		Lijian Jiang\thanks{School of Mathematical Sciences,  Tongji University, Shanghai 200092, China. ({\tt  ljjiang@tongji.edu.cn}).}
	}
\date{}
\maketitle
\begin{center}{\bf Abstract}
\end{center}\smallskip
In this article, we present a data-driven method for parametric models with noisy observation data. Gaussian process regression based  reduced order modeling (GPR-based ROM) can realize fast online predictions without using equations in the offline stage. However, GPR-based ROM does not perform well for complex systems  since POD projection are naturally linear. Conditional variational autoencoder (CVAE) can address this issue via nonlinear neural networks but it has more model complexity, which poses challenges for training and tuning hyperparameters. To this end, we propose a framework of CVAE with Gaussian process regression recognition (CVAE-GPRR). The proposed method consists of a recognition model and a  likelihood model. In the recognition model, we first extract low-dimensional features from data by POD to filter the redundant information with high frequency. And then a non-parametric model GPR is used to learn the map from parameters to POD latent variables, which can also alleviate the impact of noise.  CVAE-GPRR can achieve the similar accuracy  to CVAE but with fewer parameters. In the likelihood model, neural networks are used to reconstruct data. Besides the samples of POD latent variables and input parameters, physical variables are also added as the  inputs to make predictions in the whole physical space. This can not be achieved by either  GPR-based ROM or  CVAE.  Moreover,  the numerical results show that CVAE-GPRR may alleviate the overfitting 
issue in CVAE.

\smallskip
{\bf keywords:} Parametric models, Conditional variational autoencoder, Proper orthogonal decomposition, Gaussian process regression.

\section{Introduction}
\label{sec:Introduction}
Many science and engineering problems can be described by parametric  models, where the  parameters may characterize  material property, initial and boundary conditions,  geometrical regions, etc. Mathematical models typically describe the physics of real-world phenomena. The increasing complexity of mathematical models poses a challenge for numerical simulation. Traditional numerical  methods, such as finite element methods, finite difference methods, require the exact form of the differential equations and possibly  lead to an intractable discrete system. However, quick response of parametric  models for different parameter values is critical  for real-time applications (e.g., Bayesian inversion  or stochastic  control \cite{intro:BIP, intro:L. Ma}) and many-query  contexts (e.g., design or optimization \cite{intro:CVAE-drug design}).

To address the above issues, intrusive reduced order modeling (ROM) \cite{intro:intrusive ROM} has been widely explored   for decades, and  provides  a trade-off between modeling accuracy and computation efficiency  by carefully constructing a reduced order model of original full order model. As a popular  intrusive ROM, Garlerkin projection is implemented  on a low-dimensional  space, which is  spanned by a set of orthogonal vectors constructed by reduced basis methods (e.g., greedy algorithm \cite{intro:basis greedy algorithm}, proper orthogonal decomposition (POD)  \cite{intro:POD-Garlekin1, intro:POD-Garlekin2}). Furthermore, techniques such as  empirical interpolation method (EIM) \cite{intro:EIM} and  discrete empirical interpolation method (DEIM) \cite{intro:DEIM}  are used to obtain  affine decompositions for nonlinear problems and achieve  fast online computation for nonlinear  models with parameter dependence. The intrusive approaches depend on the underlying physics embedded in mathematical models. This limits the applications  because the mathematical models are often unavailable due to a lack of physical process  and  expert knowledge.  The intrusive ROM may  not well explain all information of  measurement  data.

 As the rapid development of science experiments  in recent years, paired (parameter-state pair) measurement data become much easier to acquire. Data-driven methods have risen in recent years to  discover latent laws of real-world problems from data. In contrast to intrusive ROM, data-driven methods can estimate the state of models  without differential  equations \cite{intro:M. L}. POD is one of the  data-driven methods  and has  been widely used to identify coherent structures in fluid \cite{intro:POD-fluid}. By singular value decomposition, POD can approximate the data information  by a linear combination of several dominant eigenvectors of covariance matrix of the data. To be applied to parametric  models, some data-driven ROMs (or non-intrusive ROMs)  achieve efficient online computation by combining regression methods (e.g., neural network \cite{intro:NN-based RB, intro:NN-based RB2, intro:NN-based RB3}, Gaussian process regression \cite{intro:GPR-based RB, intro:GPR-based RB2} and radial basis function regression \cite{intro:RBF-based RB}) with POD to learn a map from parameters to POD projection coefficients. The performance of these approaches highly depends on the number of basis functions (or vectors). However, when the data is noisy, it is difficult to determine a suitable  number of basis functions. In addition, POD basis is discontinuous with respect to the physical variables. For the unobserved physical region (i.e., no data information in the region), the ROMs need an  interpolation for prediction.

 As an unsupervised deep learning method, autoencoder (AE) \cite{intro:AE-original paper}  provides a more general framework for data-driven model reduction.  AE consists of two models: encoder and decoder.  The encoder (recognition model) extracts principle low-dimensional features  in the latent space from data and  alleviates the impact of noise and redundant information. And then the decoder (i.e., generative/likelihood model) reconstructs original data from the outputs of encoder. However, autoencoder is a deterministic method and prone to overfitting due to small datasets and a large number of  training parameters \cite{intro:overfit-AE}, and it can not well qualify the uncertainty of reconstruction data. Variational autoencoder (VAE) \cite{intro:VAE} is a probabilistic extension of autoencoder based on variational inference, which treats latent variables as random vectors and thus can model the distribution of data. By using an amortization network to approximate the posterior of the latent variables in a parametric manner, VAE can optimize recognition model and likelihood model at the same time. However, more parameters means more computational costs. In addition, VAE is a black-box model with neural networks that are usually not interpretable.

In this paper, we consider the case that only noisy observation data are available for unknown parametric  models. We want to develop   a data-driven method to model the parametric  models and make predictions for the unobserved parameter space and physical region. Compared to VAE, CVAE \cite{intro:CVAE} can generate samples from data space conditioned on some additional  information. Therefore, CVAE is more desirable for  the parametric models because we aim  at learning the distribution of states conditioned on the parameters. Combined  with data-driven ROM techniques, we want to use fewer parameters than standard CVAE. Instead of using neural networks to approximate the posterior of latent variables, we use the non-parametric regression method GPR to establish a mapping from the parameter space to the  latent space after extracting the low-dimensional features from data by POD. We show that  POD projection coefficients are mutually uncorrelated with respect to the parameters, and  learn the GPR models for each POD latent variable separately, which allows an efficient computation for high-dimensional latent spaces. To generate samples of the states in the whole physical space, we add physical variables into the inputs of the likelihood model. Since  the proposed method share the similarity to  CVAE, we call the method as conditional variational autoencoder with Gaussian process regression recognition and refer it to CVAE-GPRR in the rest of paper for simplicity.

This paper is organized as follows. In Section \ref{sec:Parametric models}, we give a brief introduction to the parametric models from the point view of deep latent variable models and   CVAE. Section \ref{sec:Proposed method} focuses on   CVAE-GPRR. In  this section, we also extensively compare the proposed  method with CVAE. In Section \ref{sec:Numerical results and discussions}, a few of numerical results are presented to illustrate the performance of the proposed method and compare it with CVAE and other data-driven ROMs. Finally, some conclusions are given.

\section{Parametric  models and CVAE}
\label{sec:Parametric models}
 In this section, we describe the general framework of parametric models and briefly present the standard conditional variational autoencoder.
\subsection{Deep latent-variable model for parametric models}
\label{ssec:Deep latent-variable model for parametric models}
Let us consider a parameter space $\Xi\subset \mathbb{R}^d$ where $d\geq 1$ with event field $\mathcal{F}_{\Omega}$ and probabilistic measure $p_{\boldsymbol{\xi}}$, and denote by $\boldsymbol{\xi}:=[\xi_1,...,\xi_d]^T$ a parameter vector encoding physical and/or geometrical properties of the problem. Furthermore, we introduce the physical variable by $\boldsymbol{x}:=[x_1,...,x_m]^T$, which is a m-dimensional vector taking values on a bounded physical region $\Omega\subset \mathbb{R}^m$. Here, the physical variable can represent space variable and/or time variable according to different application scenarios. We consider the following problems
\begin{equation}
	\mathscr{L}_{\boldsymbol{\gamma}}u(\boldsymbol{x},\boldsymbol{\xi})=0,
	\label{sec2_eq:parametric systems}
\end{equation}
whose response $u(\boldsymbol{x},\boldsymbol{\xi})$ depends on one or more parameters $\boldsymbol{\xi}$ and other invariant factors that influence the system are collected in the set $\boldsymbol{\gamma}$ such as fixed boundary condition and forcing. Such problems can be summarized as above parametric models, examples include but are not limited to parametric partial differential equations and parametric ordinary differential equations.

We assume that the response of a model $u(\boldsymbol{x},\boldsymbol{\xi})$ can be partly observed by numerical simulations or measurements. Our goal is to design a data-driven method to learn the data distribution conditioned on parameters and physical variables $p^{*}(U|\boldsymbol{s},\boldsymbol{\xi})$ without knowing mathematical equations. In this paper, the observation data for the $i-$th sample $\boldsymbol{\xi}^{(i)}$ of parameter vector is a matrix $\mathcal{D}^{(i)}\in \mathbb{R}^{M\times(m+d+1)}$ with the following structure
\begin{equation}
	\mathcal{D}^{(i)}=\begin{bmatrix}
		x^{(1)}_1&\cdots&x^{(1)}_m&\xi^{(i)}_1&\cdots&\xi^{(i)}_d&u(\boldsymbol{x}^{(1)},\boldsymbol{\xi}^{(i)})\\
		x^{(2)}_1&\cdots&x^{(2)}_m&\xi^{(i)}_1&\cdots&\xi^{(i)}_d&u(\boldsymbol{x}^{(2)},\boldsymbol{\xi}^{(i)})\\
		\vdots&\ddots&\vdots&\vdots&\ddots&\vdots&\vdots\\
		x^{(M)}_1&\cdots&x^{(M)}_m&\xi^{(i)}_1&\cdots&\xi^{(i)}_d&u(\boldsymbol{x}^{(M)},\boldsymbol{\xi}^{(i)})\\
	\end{bmatrix}, i=1,...,D.
\end{equation}
Then the total data matrix $\boldsymbol{\mathcal{D}}$ can be obtained by concatenating above matrices by row, i.e., $\boldsymbol{\mathcal{D}}=[\mathcal{D}^{(1)};...;\mathcal{D}^{(D)}]$. In the context of some discrete model reduction methods, the observation data can also be organized as a snapshot matrix $\boldsymbol{\mathcal{S}}:=[\boldsymbol{U}^{(1)};...;\boldsymbol{U}^{(D)}]$ where $\boldsymbol{U}^{(i)}\in \mathbb{R}^{1\times M}$ is an observation of $u(\boldsymbol{x}_{1:M},\boldsymbol{\xi}^{(i)})^T$ and $M$ different discrete points in physical region $\{\boldsymbol{x}^{(i)}\}_{i=1}^M$ is denoted by $\boldsymbol{x}_{1:M}$.

In the context of regression tasks, the goal is to find a prediction function $\Psi_{\theta}:\Omega\times\Xi\rightarrow
\mathbb{R}$ such that the corresponding model evidence $p_{\theta}(U|\boldsymbol{x},\boldsymbol{\xi})$ can accurately approximate the true model evidence $p^{*}(U|\boldsymbol{x},\boldsymbol{\xi})$ in Kullback–Leibler(KL) divergence
\begin{equation}
	\theta^{*}=\mathop{\arg\min}\limits D_{KL}\big(p^{*}(U|\boldsymbol{s},\boldsymbol{\xi})\Vert p_{\theta}(U|\boldsymbol{s},\boldsymbol{\xi})\big),
	\label{sec2_eq:OP_regression1}
\end{equation}
where KL divergence is a common tool to measure the  distance between two probability densities and is defined as $D_{KL}(p(x)||q(x))=\int p(x) \log\big({p(x)}/{q(x)}\big) dx$.  If we choose the likelihood model as independent Gaussian distributions (or we use squared-error loss), namely
\begin{equation}
	U=\Psi_{\theta}(\boldsymbol{x},\boldsymbol{\xi})+\epsilon, \ \epsilon\mathop{\sim}\limits^{i.i.d}\mathcal{N}(0,\sigma_{obs}^2).
	\label{sec2_eq:LE_regression}
\end{equation}
The following theorem can specifies the exact form of the optimal prediction function $\Psi^{*}_{\theta}$.
\begin{thm}{\rm(\cite{sec2:regression problems}, Chapter 2)}
	If the likelihood model of data is chosen as model (\ref{sec2_eq:LE_regression}), the optimal prediction function $\Psi^*$ equals to the conditional expectation of $U$ conditioned on $\boldsymbol{x}$ and $\boldsymbol{\xi}$
	\begin{equation}
		\Psi_{\theta^{*}}(\boldsymbol{x},\boldsymbol{\xi})=\mathbb{E}_{p^{*}(U|\boldsymbol{x},\boldsymbol{\xi})}(U|\boldsymbol{x},\boldsymbol{\xi}).
		\nonumber
	\end{equation}
\end{thm}
However, the true model evidence $p^{*}(U|\boldsymbol{x},\boldsymbol{\xi})$ is often too complex to be approximated accurately by Gaussian distribution. Some existing works introduce latent variables $\boldsymbol{Z}$ to obtain a better model, namely
\begin{equation}
	\begin{aligned}
		p_{\theta}(U|\boldsymbol{x},\boldsymbol{\xi})&=\int p_{\theta}(U|\boldsymbol{z},\boldsymbol{x},\boldsymbol{\xi})p_{\theta}(\boldsymbol{Z}|\boldsymbol{x},\boldsymbol{\xi})\ d\boldsymbol{z}\\
		&=\mathbb{E}_{p_{\theta}(\boldsymbol{Z}|\boldsymbol{x},\boldsymbol{\xi})}\bigg(p_{\theta}(U|\boldsymbol{z},\boldsymbol{x},\boldsymbol{\xi})\bigg).
	\end{aligned}
	\label{sec2_eq:LVM}
\end{equation}
It is obvious that, even when each factor $\big($prior $p_{\theta}(\boldsymbol{Z}|\boldsymbol{x},\boldsymbol{\xi})$ and conditional distribution $p_{\theta}(U|\boldsymbol{z},\boldsymbol{x},\boldsymbol{\xi})\big)$ is relatively simple such as Gaussian distribution, the approximation model evidence $p_{\theta}(U|\boldsymbol{x},\boldsymbol{\xi})$ can be very complex.

Above all, the optimization problem (\ref{sec2_eq:OP_regression1}) can be rewritten as
\begin{equation}
	\begin{aligned}	
		\theta^{*}&=\mathop{\arg\max}\limits \mathbb{E}_{p^{*}(U|\boldsymbol{x},\boldsymbol{\xi})}\bigg(\log\big(\mathbb{E}_{p_{\theta}(\boldsymbol{Z}|\boldsymbol{x},\boldsymbol{\xi})}p_{\theta}(U|\boldsymbol{z},\boldsymbol{x},\boldsymbol{\xi})\big)\bigg).\\
	\end{aligned}
	\label{sec2_eq:OP_LVM}
\end{equation}
Furthermore, the probability model (\ref{sec2_eq:LVM}) is well-known deep latent-variable model (DLVM) \cite{sec2:reparametrization} if neural networks are used to learn the parameters in likelihood model $p_{\theta}(U|\boldsymbol{z},\boldsymbol{x},\boldsymbol{\xi})$ where $U=\Psi_{\theta}(\boldsymbol{Z},\boldsymbol{x},\boldsymbol{\xi})+\epsilon$. In the next section, we will introduce a class of DLVM — Conditional variational autoencoder (CVAE), which is the foundation of our work.

\subsection{Conditional variational autoencoders}
\label{ssec:Conditional variational autoencoders}
In previous section, we introduced deep latent-variable models. It is usually intractable to solve the maximum likelihood problem defined in (\ref{sec2_eq:OP_LVM}) since the model evidence $p_{\theta}(U|\boldsymbol{x},\boldsymbol{\xi})$ does not have an analytic solution or efficient estimator. The framework of (conditional) variational autoencoder provides a computationally efficient way to address this problem with approximate inference techniques for approximating the posterior $p_{\theta}(\boldsymbol{Z}|U,\boldsymbol{x},\boldsymbol{\xi})$ in DLVM. According to Bayes's rule, we have
\begin{equation} p_{\theta}(\boldsymbol{Z}|U,\boldsymbol{x},\boldsymbol{\xi})=\frac{p_{\theta}(U,\boldsymbol{Z}|\boldsymbol{x},\boldsymbol{\xi})}{p_{\theta}(U|\boldsymbol{x},\boldsymbol{\xi})}.
	\nonumber
\end{equation}
A tractable posterior thus leads to a tractable model evidence since complete likelihood model $p_{\theta}(U,\boldsymbol{Z}|\boldsymbol{x},\boldsymbol{\xi})$ is tractable to compute. In this section, we consider the case that the data are observed at $M$ discrete points in $\Omega$, and then the posterior does not depend on the physical variables $\boldsymbol{x}$, i.e., $p_{\theta}(\boldsymbol{Z}|\boldsymbol{U},\boldsymbol{\xi})$.

A family of posterior $q_{\phi}(\boldsymbol{Z}|\boldsymbol{U},\boldsymbol{\xi})$ is introduced to approximate the posterior of latent variables in KL divergence as follows
\begin{equation}
	\phi^{*}=\mathop{\arg\min}\limits D_{KL}\big(q_{\phi}(\boldsymbol{Z}|\boldsymbol{U},\boldsymbol{\xi})\Vert p_{\theta}(\boldsymbol{Z}|\boldsymbol{U},\boldsymbol{\xi})\big).
	\label{sec2:KL_CVAE}
\end{equation}
It is impossible to solve above optimization problem directly because the true posterior is unknown. To get rid of this term, we perform some algebraic manipulations and arrive at
\begin{equation}
	\begin{aligned}
		D_{KL}\big(q_{\phi}(\boldsymbol{Z}|\boldsymbol{U},\boldsymbol{\xi})\Vert p_{\theta}(\boldsymbol{Z}|\boldsymbol{U},\boldsymbol{\xi})\big)&=\int q_{\phi}(\boldsymbol{Z}|\boldsymbol{U},\boldsymbol{\xi})\log\frac{q_{\phi}(\boldsymbol{Z}|\boldsymbol{U},\boldsymbol{\xi})}{p_{\theta}(\boldsymbol{Z}|\boldsymbol{U},\boldsymbol{\xi})}d\boldsymbol{z}\\
		&=-\int q_{\phi}(\boldsymbol{Z}|\boldsymbol{U},\boldsymbol{\xi})\log\frac{p_{\theta}(\boldsymbol{U},\boldsymbol{Z}|,\boldsymbol{\xi})}{q_{\phi}(\boldsymbol{Z}|\boldsymbol{U},\boldsymbol{\xi})p_{\theta}(\boldsymbol{U}|,\boldsymbol{\xi})}d\boldsymbol{z}\\
		&=-\mathbb{E}_{q_{\phi}}\bigg(\log\frac{p_{\theta}(\boldsymbol{U},\boldsymbol{Z}|\boldsymbol{\xi})}{q_{\phi}(\boldsymbol{Z}|\boldsymbol{U},\boldsymbol{\xi})}\bigg)+\log  p_{\theta}(\boldsymbol{U}|\boldsymbol{\xi}).\\
	\end{aligned}
	\label{sec2:deduction_ELBO}
\end{equation}
Since the model evidence $p_{\theta}(\boldsymbol{U}|\boldsymbol{\xi})$ is constant with respect to parameters $\phi$, minimizing the KL divergence in (\ref{sec2:KL_CVAE}) is equivalent to maximizing the first term of the right-hand side of equation (\ref{sec2:deduction_ELBO}). Furthermore, reorganizing the last equation in (\ref{sec2:deduction_ELBO}), we achieve
\begin{equation}
	\log  p_{\theta}(\boldsymbol{U}|\boldsymbol{\xi})=\mathbb{E}_{q_{\phi}}\bigg(\log\frac{p_{\theta}(\boldsymbol{U},\boldsymbol{Z}|\boldsymbol{\xi})}{q_{\phi}(\boldsymbol{Z}|\boldsymbol{U},\boldsymbol{\xi})}\bigg)+D_{KL}\big(q_{\phi}(\boldsymbol{Z}|\boldsymbol{U},\boldsymbol{\xi})\Vert p_{\theta}(\boldsymbol{Z}|\boldsymbol{U},\boldsymbol{\xi})\big).
	\label{sec2:deduction_ELBO2}
\end{equation}
By the fact that $D_{KL}\big(q_{\phi}(\boldsymbol{Z}|\boldsymbol{U},\boldsymbol{\xi})\Vert p_{\theta}(\boldsymbol{Z}|\boldsymbol{U},\boldsymbol{\xi})\big)\geq 0$, the expectation in equation (\ref{sec2:deduction_ELBO2}) is a lower bound on $\log p_{\theta}(\boldsymbol{U}|\boldsymbol{\xi})$. For this reason, it is called the evidence lower bound (ELBO). We can rearrange the ELBO into the more interpretable form
\begin{equation} ELBO(q)=\mathbb{E}_{q_{\phi}}\bigg(p_{\theta}(\boldsymbol{U}|\boldsymbol{z},\boldsymbol{\xi})\bigg)-D_{KL}\big(q_{\phi}(\boldsymbol{Z}|\boldsymbol{U},\boldsymbol{\xi})\Vert p_{\theta}(\boldsymbol{Z}|\boldsymbol{\xi})\big).
	\label{sec2:CVAE_ELBO}
\end{equation}
The first term is the expected likelihood under the posterior of latent variables, which can let the likelihood model better explain the data. The second term acts as a regularizer to push the posterior of latent variables toward the prior. It can be observed that there are two distributions, the  posterior of latent variables $q_{\phi}(\boldsymbol{Z}|\boldsymbol{U},\boldsymbol{\xi})$ and the likelihood $p_{\theta}(\boldsymbol{U}|\boldsymbol{z},\boldsymbol{\xi})$, to be approximated in equation (\ref{sec2:CVAE_ELBO}). Neural networks are used to learn the parameters in distributions as follows \cite{sec2:CVAE}:
\begin{equation}
	\phi=NN_{recog}(\boldsymbol{U},\boldsymbol{\xi};\Phi),
	\label{sec2:CVAE_recognition}
\end{equation}
the local parameters $\phi$ in the approximated posterior $q_{\phi}(\boldsymbol{Z}|\boldsymbol{U},\boldsymbol{\xi})$ varies with data. With above amortized inference techniques, recognition model (\ref{sec2:CVAE_recognition}) can share the global parameters $\Phi$ across data points.
\begin{equation}
	\boldsymbol{z}^{(i)}=g_{\Phi}(\boldsymbol{U},\epsilon_0^{(i)}),\ \epsilon_0^{(i)}\sim p(\epsilon_0)\  \text{(base distribution)}.
	\label{sec2:CVAE_reparametrization}
\end{equation}
And then we can draw samples of latent variables $Z$ from posterior by reparameterization trick (\ref{sec2:CVAE_reparametrization}), which enables us to use gradient descent to optimize the model.
Next, we pass the sample $\boldsymbol{z}^{(i)}$ through the likelihood model (\ref{sec2:CVAE_likelihood}) to obtain the parameters of distribution $p_{\theta}(\boldsymbol{U}|\boldsymbol{z},\boldsymbol{\xi})$. 
\begin{equation}
	\theta=NN_{le}(\boldsymbol{z}^{(i)},\boldsymbol{\xi};\Theta).
	\label{sec2:CVAE_likelihood}
\end{equation}
\section{Conditional variational autoencoder with Gaussian process regression recognition (CVAE-GPRR)}
\label{sec:Proposed method}
Similar to CVAE, the proposed method also consists of two separate models: recognition model and likelihood model. Inspired by data-driven ROM techniques, we obtain a more interpretable recognition model with fewer parameters than the standard CVAE. We first use POD to filter redundant information that is high frequency from data to obtain low-dimensional observations of latent variables. And then non-parametric regression method GPR is used to learn the mapping from parameter space to latent space, which also helps to denoise. For likelihood model, neural networks are used to reconstruct data from latent space since POD basis may not be the optimal choice when data is noisy. By adding physical variables as parts of inputs, our framework can generate samples in the unobserved physical region. In addition, we will extensively compare CVAE-GPRR with CVAE in the last subsection.

\subsection{GPR recognition model}
\label{ssec:GPR recognition model}
POD in CVAE-GPRR is used to find the underlying low-dimensional structure from data and the projection coefficients are chosen as latent variables to make the model evidence more complex. And then GPR is used to learn the mapping from parameter space to latent space.

\subsubsection{The proper orthogonal decomposition and POD latent variables}
\label{ssec:The proper orthogonal decomposition and POD-latent variables}
Let $\mathcal{H}(\Omega)$ be a separable Hilbert space with inner product $\left<\cdot,\cdot\right>_{\mathcal{H}(\Omega)}$ and induced norm $\Vert\cdot\Vert_{\mathcal{H}(\Omega)}$ , which takes values in physical region $\Omega$. For any subspace $\mathcal{V}\subseteq\mathcal{H}(\Omega)$, define an orthogonal projection operator $\Pi_{\mathcal{V}}:\mathcal{H}(\Omega)\rightarrow \mathcal{V}$, and then the projection error can be written as
\begin{equation}
	\mathcal{R}(\mathcal{V}):=\mathbb{E}_{p_{(\boldsymbol{\xi})}}\Vert u(\boldsymbol{x},\boldsymbol{\xi})-\Pi_{\mathcal{V}}u(\boldsymbol{x},\boldsymbol{\xi})\Vert_{\mathcal{H}(\Omega)}^2.
	\nonumber
\end{equation}\par
The main idea of POD is to find a finite subspace $\mathcal{V}_k:=span\{v_1,v_2,...,v_k\}$ of $\mathcal{H}(\Omega)$ to minimize the projection error $\mathcal{R}(\mathcal{V}_{k})$. In terms of the definition of $\mathcal{V}_k$, the orthogonal projection operator can be written as
\begin{equation}
	\Pi_{\mathcal{V}_k}u(\boldsymbol{x},\boldsymbol{\xi})=\sum_{i=1}^k\left<u(\boldsymbol{x},\boldsymbol{\xi}),v_i\right>_{\mathcal{H}(\Omega)}v_i.
	\nonumber
\end{equation} \par
Given a snapshot matrix of data $\boldsymbol{\mathcal{S}} \in \mathbb{R}^{D\times M}$ that is defined in Subsection \ref{ssec:Deep latent-variable model for parametric models}, we can define empirical projection error by Monte Carlo and obtain the discrete form of orthogonal projection operator, namely
\begin{equation}
	\mathcal{R}_{D,M}(V_{k,M}):=\frac{1}{D}\sum_{i=1}^{D}\Vert \boldsymbol{U}^{(i)}-\Pi_{V_{k,M}}\boldsymbol{U}^{(i)}\Vert_2^2,
	\nonumber
\end{equation}
where the norm $\Vert\cdot\Vert_2$ is the vector 2-norm and discrete projection operator is defined as $\Pi_{V_{k,M}}\boldsymbol{U}^{(i)}:=V_{k,M}V_{k,M}^T[\boldsymbol{U}^{(i)}]^T, \ V_{k,M}=[\boldsymbol{v}_1,...,\boldsymbol{v}_k]\in R^{M\times k}$. Furthermore, we define the empirical covariance matrix
\begin{equation}
	\hat{C}_{D,u}=\frac{1}{D}\sum_{i=1}^{D}(\boldsymbol{U}^{(i)}-\overline{\boldsymbol{\mathcal{D}}})^T(\boldsymbol{U}^{(i)}-\overline{\boldsymbol{\mathcal{D}}}),\ \overline{\boldsymbol{\mathcal{D}}}=\frac{1}{D}\sum_{i=1}^{D}\boldsymbol{U}^{(i)}.
	\label{sec3:empirical covariance matrix}
\end{equation}
It is well known that $\frac{D}{D-1}\hat{C}_{D,u}$ is an unbiased estimation of covariance matrix
\begin{equation} C_{u}=\mathbb{E}_{p_{(\boldsymbol{\xi})}}\bigg(\big(u(\boldsymbol{x}_{1:M},\boldsymbol{\xi})-\mathbb{E}_{p_{(\boldsymbol{\xi})}}u(\boldsymbol{x}_{1:M},\boldsymbol{\xi})\big)^T\big(u(\boldsymbol{x}_{1:M},\boldsymbol{\xi})-\mathbb{E}_{p_{(\boldsymbol{\xi})}}u(\boldsymbol{x}_{1:M},\boldsymbol{\xi})\big)\bigg).
	\nonumber
\end{equation}
Suppose $\hat{C}_{D,u}=V\Lambda^2V^T$ is an eigenvalue decomposition of $\hat{C}_{D,u}$. Then the columns of POD basis matrix $V_{k,M}$, which solves optimization problem $\mathop{\min}\limits_{V_{k,M}}\mathcal{R}_{D,M}(V_{k,M})$, are eigenvectors corresponding to the $k$ largest diagonal elements in $\Lambda^2:=diag(\lambda_1,...,\lambda_M),\ \lambda_1\ge...\ge\lambda_M\ge0$. Furthermore,
\begin{equation}
	\mathcal{R}_{D,M}(V_{k,M})=\sum_{i=k+1}^{M}\lambda_i,
	\nonumber
\end{equation}
from which we can infer that the square of relative empirical projection error is
\begin{equation}
 E_k:=\frac{\sum_{i=k+1}^{M}\lambda_i}{\sum_{i=1}^{M}\lambda_i}.
 \nonumber
\end{equation}
By setting a tolerance $\epsilon_{POD}$, we can quickly select a suitable hyperparameter $k$, that is
\begin{equation}
	k:=\mathop{\arg\min}\{0\le k\le M |E_k \le \epsilon_{POD}^2\}.
	\label{sec3:pod_k}
\end{equation}
The procedure of POD is summarized in the following algorithm.
\begin{algorithm}[H]
	\caption{Proper Orthogonal Decomposition (POD)} 
	\hspace*{0.02in} \textbf{Input:} Snapshot matrix $\boldsymbol{\mathcal{D}}$, relative empirical projection error tolerance $\epsilon_{POD}$\\
	\hspace*{0.02in} {\bf Output:} 
	POD basis matrix $V_{k,M}$
	\begin{algorithmic}[1]
		\State Centering the data by subtracting the estimated mean with respect to $\boldsymbol{\xi}$: $\hat{\boldsymbol{\mathcal{D}}}=\boldsymbol{\mathcal{D}}-\overline{\boldsymbol{\mathcal{D}}}$, where the computation of $\overline{\boldsymbol{\mathcal{D}}}$ refers to equation (\ref{sec3:empirical covariance matrix});
		\State Compute the covariance matrix $\hat{C}_{D,u}$ estimated from data  according to equation (\ref{sec3:empirical covariance matrix});
		\State Solve the eigenvalue problem for $\hat{C}_{D,u}$, i.e., $\hat{C}_{D,u} \boldsymbol{v}_i=\lambda_i \boldsymbol{v}_i, \ i=1,...,M$;
		\State Select hyperparameter $k$ according to equation (\ref{sec3:pod_k});
		\State \Return POD basis matrix $V_{k,M}:=[\boldsymbol{v}_1,...,\boldsymbol{v}_k]$.
	\end{algorithmic}
	\label{sec3:algorithm_POD}
\end{algorithm}
The discrete orthogonal projection operator $\Pi_{V_{k,M}}$ can be seen as a linear combination of POD basis $\boldsymbol{v}_i$. Therefore, we rewritten it as follows
\begin{equation}
	\Pi_{V_{k,M}}(\boldsymbol{U}^{(i)})=\sum_{j=1}^k\left<[\boldsymbol{U}^{(i)}]^T,\boldsymbol{v}_j\right>_2\boldsymbol{v}_j=V_{k,M}\boldsymbol{z}_i,
	\label{sec3:pod-latent variables}
\end{equation}
where $\boldsymbol{z}_i=V_{k,M}^T[\boldsymbol{U}^{(i)}]^T, \boldsymbol{z}_i=[z_{i,1},...,z_{i,k}]^T, z_{i,j}=\left<[\boldsymbol{U}^{(i)}]^T,\boldsymbol{v}_j\right>_2$ and $\left<\cdot,\cdot\right>_2$ is the Euclidean inner product. We refer $\boldsymbol{Z}=[\boldsymbol{Z}_1,...,\boldsymbol{Z}_k]$ as POD projection coefficients and choose it as latent variables for model (\ref{sec2_eq:LVM}). In the context of DLVM, we use the term POD latent variables to denote the random variables $\boldsymbol{Z}$. For each data snapshot $\boldsymbol{U}^{(i)}$, $z_{i,j}$ can be seen as a realization for the j-th entry of $\boldsymbol{Z}$. It is obvious that $\boldsymbol{Z}$ varies according to $\boldsymbol{\xi}$, that is, there is a multioutput mapping $\boldsymbol{\pi}$ such that $\boldsymbol{Z}=\boldsymbol{\pi}_{\boldsymbol{Z}}(\boldsymbol{\xi})$. In the following subsections, Gaussian process regression is utilized to learn the mapping $\pi$ and thus obtain the posterior for $\boldsymbol{Z}$ conditioned on $\boldsymbol{\xi}$ and data $\boldsymbol{\mathcal{D}}$.

\subsubsection{Gaussian process regression}
Gaussian process regression is one of the most popular Bayesian non-parametric regression techniques based on kernel machinery. The goal of GPR is to learn a regression function $\pi_{GPR}$ that predicts the output $y=\pi_{GPR}(\boldsymbol{\xi})\in \mathbb{R}$ given an input $\boldsymbol{\xi}$.

First, the prior distribution of $\pi_{GPR}$ is given by Gaussian processes with mean function $\mu(\boldsymbol{\xi})$ and covariance function $\kappa(\boldsymbol{\xi},\boldsymbol{\xi}^{'})$. For simplicity, we take $\mu$ to be zero. And $\kappa$ is a kernel function that can describe the correlation between inputs. For instance, radial basis function kernels (including squared exponential(SE) kernel) suppose that points that are closer will be more highly correlated. The kernel function, that we will use in this paper, is the automatic relevance determination squared exponential kernel (ARD SE kernel) \cite{sec3:ARD SE kernel}:
\begin{equation}
	\kappa(\boldsymbol{\xi},\boldsymbol{\xi}^{'})=\sigma_{\pi}^2\exp\bigg(-\frac{1}{2}\sum_{i=1}^d\frac{(\xi_i-\xi_i^{'})^2}{l_i^2}\bigg),
	\label{sec3:kernel function}
\end{equation}
which weights each input dimension by individual lengthscale $l_i$ compared to the squared exponential kernel.

Second, we regard the observed input $\{\boldsymbol{\xi}^{(i)}\}_{i=1}^{D_{GPR}}$ as determinant and the output \\$\{y_i\}_{i=1}^{D_{GPR}}$ as outcomes of random variables $\{Y_i\}_{i=1}^{D_{GPR}}$. Therefore, given prediction function $\pi_{GPR}$, we can model $\{Y_i\}_{i=1}^{D_{GPR}}$ with Gaussian noise as follows
\begin{equation}
	Y_i=\pi_{GPR}(\boldsymbol{\xi}^{(i)})+\epsilon_{GPR,i},\ i=1,...,D_{GPR},
	\nonumber
\end{equation}
where $\epsilon_{GPR,i}\mathop{\sim}\limits^{i.i.d}\mathcal{N}(0,\sigma_{GPR}^2)$. We denote the regression values by
\begin{equation}
 \boldsymbol{\pi}_{GPR}=[\pi_{GPR}(\boldsymbol{\xi}^{(1)}),...,\pi_{GPR}(\boldsymbol{\xi}^{(D_{GPR})})]^T,
 \nonumber
\end{equation}
and denote the hyperparameters by $\alpha:=\{\sigma_{\pi},l_1,...,l_d,\sigma_{GPR}\}$.
According to the definition of Gaussian process, we know that GP prior over function $\pi$ is equivalent to placing a Gaussian prior on the vector $\boldsymbol{\pi}_{GPR}$:
\begin{equation}
	\boldsymbol{\pi}_{GPR}\sim\mathcal{N}(0,\boldsymbol{K}),
	\label{sec3:GPR prior}
\end{equation}
where the covariance matrix $\boldsymbol{K}$ is a Gram matrix,namely
\begin{equation}
	\boldsymbol{K}=\begin{bmatrix}
		\kappa(\boldsymbol{\xi}^{(1)},\boldsymbol{\xi}^{(1)})& \kappa(\boldsymbol{\xi}^{(1)},\boldsymbol{\xi}^{(2)})&\cdots
		&\kappa(\boldsymbol{\xi}^{(1)},\boldsymbol{\xi}^{(D_{GPR})})\\
		\kappa(\boldsymbol{\xi}^{(2)},\boldsymbol{\xi}^{(1)})& \kappa(\boldsymbol{\xi}^{(2)},\boldsymbol{\xi}^{(2)})&\cdots
		&\kappa(\boldsymbol{\xi}^{(2)},\boldsymbol{\xi}^{(D_{GPR})})\\
		\vdots&\vdots&\ddots&\vdots\\
		\kappa(\boldsymbol{\xi}^{(D_{GPR})},\boldsymbol{\xi}^{(1)})& \kappa(\boldsymbol{\xi}^{(D_{GPR})},\boldsymbol{\xi}^{(2)})&\cdots
		&\kappa(\boldsymbol{\xi}^{(D_{GPR})},\boldsymbol{\xi}^{(D_{GPR})})
	\end{bmatrix}.
\nonumber
\end{equation}
Then the likelihood of observation is $\boldsymbol{Y}|\boldsymbol{\pi}_{GPR},\boldsymbol{\xi}_{1:D_{GPR}}\sim\mathcal{N}(\boldsymbol{\pi},\sigma_{GPR}^2\boldsymbol{I}_{D_{GPR}})$. Combined with the prior (\ref{sec3:GPR prior}), we can obtain the marginal distribution through computation:
\begin{equation}
	\boldsymbol{Y}|\boldsymbol{\xi}_{1:D_{GPR}}\sim\mathcal{N}(0,\boldsymbol{K}_{\boldsymbol{y}})
	\nonumber
\end{equation}
where $\boldsymbol{\xi}_{1:D_{GPR}}=\{\boldsymbol{\xi}^{(i)}\}_{i=1}^{D_{GPR}},\ \boldsymbol{Y}=[Y_1,Y_2,...,Y_{D_{GPR}}]$ and $\boldsymbol{K}_{\boldsymbol{y}}=\boldsymbol{K}+\sigma_{GPR}^2\boldsymbol{I}_{D_{GPR}}$. Given a new input $\boldsymbol{\xi}^{*}$, denote the corresponding output by $\pi^{*}$. We can estimate $\pi^{*}$ from the posterior $q_{\phi^{*}}(\pi^{*}|\boldsymbol{y},\boldsymbol{\xi}^{*};\alpha)$, which is also a Gaussian distribution with mean and variance given by
\begin{equation}
	\mu(\boldsymbol{\xi}^{*})=\boldsymbol{\kappa}(\boldsymbol{\xi}^{*})^T\boldsymbol{K}_{\boldsymbol{y}}^{-1}\boldsymbol{y}
	\label{sec3:GPR mean}
	\nonumber
\end{equation}
and
\begin{equation}
	\sigma(\boldsymbol{\xi}^{*})=\kappa(\boldsymbol{\xi}^{*},\boldsymbol{\xi}^{*})-\boldsymbol{\kappa}^T\boldsymbol{K}_{\boldsymbol{y}}^{-1}\boldsymbol{\kappa},
	\label{sec3:GPR var}
	\nonumber
\end{equation}
where we collect $\mu(\boldsymbol{\xi}^{*})$ and $\sigma({\boldsymbol{\xi}}^{*})$ into vector $\phi^{*}: =[\mu(\boldsymbol{\xi}^{*}),\sigma(\boldsymbol{\xi}^{*})]$ ,$\boldsymbol{y}=[y_1,...,y_{D_{GPR}}]^T$ and $\boldsymbol{\kappa}(\boldsymbol{\xi}^{*})=[\kappa(\boldsymbol{\xi}^{*},\boldsymbol{\xi}^{(1)}),...,\kappa(\boldsymbol{\xi}^{*},\boldsymbol{\xi}^{(D_{GPR})})]^T$. The posterior can be obtained via Bayes's rules $q_{\phi^{*}}(\pi^{*}|\boldsymbol{y},\boldsymbol{\xi}^{*};\alpha)={q(\boldsymbol{y},\pi|\boldsymbol{\xi}_{1:D_{GPR}},\boldsymbol{\xi}^{*};\alpha)}/{q(\boldsymbol{y}|\boldsymbol{\xi}_{1:D_{GPR}};\alpha)}$ where the joint distribution is also Gaussian with zero mean and covariance matrix $\boldsymbol{K}_{\boldsymbol{y},\pi^{*}}$
\begin{equation}
	\boldsymbol{K}_{\boldsymbol{y},\pi^{*}}=\begin{bmatrix}
		\boldsymbol{K}_{\boldsymbol{y}}&\boldsymbol{\kappa}(\boldsymbol{\xi}^{*})\\
		\boldsymbol{\kappa}(\boldsymbol{\xi}^{*})^T&\kappa(\boldsymbol{\xi}^{*},\boldsymbol{\xi}^{*})
	\end{bmatrix}.
\nonumber
\end{equation}\par
The prediction of GPR depends ,to some degree, on the hyperparameters $\alpha$. Here, we make a point estimate of $\alpha$ from data by maximizing the log likelihood function. Using a gradient-based optimizer, we can get the optimal $\alpha^{*}$ :
\begin{equation}
	\alpha^{*}=\mathop{\arg\max}\limits_{\alpha\ge 0}\ \log q(\boldsymbol{y}|\boldsymbol{\xi}_{1:D_{GPR}};\alpha)=\mathop{\arg\max}\limits_{\alpha\ge 0}\ -\frac{1}{2}\log\left|\det(\boldsymbol{K}_{\boldsymbol{y}})\right|-\frac{1}{2}\boldsymbol{y}^T\boldsymbol{K}_{\boldsymbol{y}}^{-1}\boldsymbol{y}.
	\label{sec3:GPR hyperparameters}
\end{equation}
The procedure of GPR is summarized in the following algorithm.
\begin{algorithm}[H]
	\caption{Gaussian process regression (GPR)} 
	\hspace*{0.02in} \textbf{Input:} A dataset $\{(\boldsymbol{\xi}^{(i)},y_i)\}_{i=1}^{D_{GPR}}$, ADR-SE kernel function $\kappa(\cdot,\cdot)$, test inputs $\boldsymbol{\Xi}_{test}^{*}:=\{\boldsymbol{\xi}_i^{*}\}_{i=1}^{D_{GPR}^{*}}$\\
	\hspace*{0.02in} {\bf Output:} 
	test mean vector $\boldsymbol{\mu}:=[\mu(\boldsymbol{\xi}_1^{*}),...,\mu(\boldsymbol{\xi}_{D_{GPR}^{*}}^{*})]^T$ and variance vector $\boldsymbol{\sigma}:=[\sigma(\boldsymbol{\xi}_1^{*}),...,\sigma(\boldsymbol{\xi}_{D_{GPR}^{*}}^{*})]^T$
	\begin{algorithmic}[1]
		\State Obtain the optimal hyperparameters $\alpha^{*}$ via the maximization problem (\ref{sec3:GPR hyperparameters});
		\State Form the covariance matrix $\boldsymbol{K}_{\boldsymbol{y}}=\boldsymbol{K}+\sigma_{GPR}^2\boldsymbol{I}_{D_{GPR}}$ and compute its inverse matrix $\boldsymbol{K}_{\boldsymbol{y}}^{-1}$;
		\State Form the matrix/vector $\boldsymbol{\kappa}:=[\boldsymbol{\kappa}({\boldsymbol{\xi}_1^{*}}),...,\boldsymbol{\kappa}({\boldsymbol{\xi}_{D_{GPR}^{*}}^{*}})]\in \mathbb{R}^{D_{GPR}\times D_{GPR}^{*}}$ and compute the test mean vector $\boldsymbol{\mu}=\boldsymbol{\kappa}^T\boldsymbol{K}_{\boldsymbol{y}}^{-1}\boldsymbol{y}$;
		\State Form the vector $\boldsymbol{\kappa}^{*}=[\kappa(\boldsymbol{\xi}_1^{*},\boldsymbol{\xi}_1^{*}),...,\kappa(\boldsymbol{\xi}_{D_{GPR}^{*}}^{*},\boldsymbol{\xi}_{D_{GPR}^{*}}^{*})]^T$ and compute the variance vector $\boldsymbol{\sigma}=\boldsymbol{\kappa}^{*}-\text{diag\_part}(\boldsymbol{\kappa}^T\boldsymbol{K}_{\boldsymbol{y}}^{-1}\boldsymbol{\kappa})$ where $\text{diag\_part}$ means taking the diagonal elements of matrix $\boldsymbol{\kappa}^T\boldsymbol{K}_{\boldsymbol{y}}^{-1}\boldsymbol{\kappa}$;
		\State \Return test mean vector $\boldsymbol{\mu}$ and variance vector $\boldsymbol{\sigma}$.
	\end{algorithmic}
	\label{sec3:algorithm_GPR}
\end{algorithm}

\subsubsection{Gaussian process regression for POD latent variables}
\label{ssec:Gaussian process regression for POD latent variables}
In this subsection, we denote the prediction function from $\Xi$ to $\mathbb{R}^k$ by $\boldsymbol{\pi}_{\boldsymbol{Z}}:\boldsymbol{\xi}\rightarrow \boldsymbol{Z}$. The following theorem can ensure that the regression can be done parallel for each entry of $\boldsymbol{Z}$.
\begin{thm}
	Each entry of POD latent variables defined in equation (\ref{sec3:pod-latent variables}) are uncorrelated under the distribution of parameter vector $\boldsymbol{\xi}$.
	\label{sec3:thm:independent pod latent variables}
\end{thm}
\begin{proof}
	To show that each entry of $\boldsymbol{Z}$ is uncorrelated, we only need to compute the covariance matrix $C_{\boldsymbol{Z}}=\mathbb{E}_{p_{(\boldsymbol{\xi})}}\bigg(\big(\boldsymbol{Z}-\mathbb{E}_{p_{(\boldsymbol{\xi})}}(\boldsymbol{Z})\big)\big(\boldsymbol{Z}-\mathbb{E}_{p_{(\boldsymbol{\xi})}}(\boldsymbol{Z})\big)^T\bigg)$ of $\boldsymbol{Z}$. We first compute the empirical covariance matrix $\hat{C}_{D,\boldsymbol{Z}}$,
	\begin{equation}
		\begin{aligned}
			\hat{C}_{D,\boldsymbol{Z}}&=\frac{1}{D}\sum_{i=1}^{D}(\boldsymbol{z}_i-\overline{\boldsymbol{Z}})(\boldsymbol{z}_i-\overline{\boldsymbol{Z}})^T\\
			&=\frac{1}{D}\sum_{i=1}^{D}\big(V_{k,M}^T(\boldsymbol{U^{(i)}})^T-V_{k,M}^T\overline{\boldsymbol{\mathcal{D}}}\big)\big(P_{k,N}^T(\boldsymbol{U^{(i)}})^T-P_{k,p}^T\overline{\boldsymbol{\mathcal{D}}}\big)^T\\
			&=P_{k,N}^T\hat{C}_{D,u}P_{k,N}=P_{k,N}^TP_{k,N}\Lambda^2=\Lambda^2,
		\end{aligned}
	\nonumber
	\end{equation}
	where $\overline{\boldsymbol{Z}}=\frac{1}{D}\sum_{i=1}^{D}\boldsymbol{z}_i$. Then we use the fact that $\frac{D}{D-1}\hat{C}_{D,\boldsymbol{Z}}$ is an unbiased estimation of $C_{\boldsymbol{Z}}$, i.e.,
	\begin{equation}
		C_{\boldsymbol{Z}}=\mathbb{E}_{p_{(\boldsymbol{\xi})}}\bigg(\frac{D}{D-1}\hat{C}_{D,\boldsymbol{Z}}\bigg)
		=\frac{D}{D-1}\Lambda^2.
	\nonumber
	\end{equation}
	We complete the proof since the eigenvalue matrix $\Lambda^2$ is a diagonal matrix.
\end{proof}
We know that uncorrelation is equivalent to independence for Gaussian distribution. Therefore, according to theorem \ref{sec3:thm:independent pod latent variables}, We can construct k single output GPR models $\pi_{z_j},\ j=1,...,k$ with dataset $\{\boldsymbol{\xi}^{(i)},z_{i,j}\}_{i=1}^{D}$ parallel instead of doing a multiple output GPR task $\boldsymbol{\pi}_{\boldsymbol{Z}}$.\\
The discrete form of data likelihood (\ref{sec2_eq:LE_regression}) is
\begin{equation}
	\boldsymbol{U}^{(i)}=\Psi_{\theta}(\boldsymbol{\xi}^{(i)})+\boldsymbol{\epsilon}_i,\ \boldsymbol{\epsilon}_i \sim \mathcal{N}(0,\sigma_{obs}^2\boldsymbol{I}_p),
	\label{sec3_eq:LE_discrete}
\end{equation}
where $\Psi_{\theta}(\boldsymbol{\xi})=\Psi_{\theta}(\boldsymbol{\xi};\boldsymbol{s}_{1:N})\in \mathbb{R}^{N\times 1}$. We can construct the likelihood model for latent variables $\boldsymbol{Z}$ from equation (\ref{sec3_eq:LE_discrete}) by POD projection
\begin{equation}
	\boldsymbol{z}_i=\boldsymbol{\pi}_{\boldsymbol{Z}}(\boldsymbol{\xi}^{(i)})+\boldsymbol{\epsilon}_i,\ \boldsymbol{\epsilon}_i \sim \mathcal{N}(0,\sigma_{obs}^2\boldsymbol{I}_k).
	\nonumber
\end{equation}
The noise level does not change since POD basis matrix is orthogonal. Then, the likelihood for each entry of $\boldsymbol{Z}$ is $z_{i,j}=\pi_{z_j}(\boldsymbol{\xi}^{(i)})+\epsilon_i, \ \epsilon_i\sim \mathcal{N}(0,\sigma_{obs}^2)$.
The GPR for POD latent variables can be summarized in the following algorithm.
\begin{algorithm}[H]
	\caption{GPR for POD latent variables} 
	\hspace*{0.02in} \textbf{Input:} A dataset $\{(\boldsymbol{\xi}^{(i)},\boldsymbol{z}_{i})\}_{i=1}^{D}$, ADR-SE kernel function $\kappa(\cdot,\cdot)$, parameter set $\boldsymbol{\Xi}_{test}^{*}:=\{\boldsymbol{\xi}_i^{*}\}_{i=1}^{D^{*}}$\\ 
	\hspace*{0.02in} {\bf Output:} 
	GPR for POD latent variables
	\begin{algorithmic}[1]
		\State Split the dataset into $k$ subset $\{(\boldsymbol{\xi}^{(i)},z_{i,j})\}_{i=1}^{D}, \ j=1,...,k$
		\State Implement Algorithm {\ref{sec3:algorithm_GPR}} with input $\{(\boldsymbol{\xi}^{(i)},z_{i,j})\}_{i=1}^{D},\ \kappa(\cdot,\cdot)$ and $\boldsymbol{\Xi}_{test}^{*}$ for latent variables $Z_j$ parallel. Then we can obtain the mean vectors $\boldsymbol{\mu}_j$ and variance vectors $\boldsymbol{\sigma}_j$.
		\State Form the mean matrix and covariance matrix for training data: $\boldsymbol{M}=[\boldsymbol{\mu}_1,...,\boldsymbol{\mu}_k]$ and $\boldsymbol{\Sigma}=[\boldsymbol{\sigma}_1,...,\boldsymbol{\sigma}_k]$ where the i-th row of $\boldsymbol{\Sigma}$ is the diagonal elements of covariance matrix for the i-th random vector $Z_i$.
		\State \Return GPR for POD latent variables $Z$ conditioned on parameter $\boldsymbol{\xi}$:a Gaussian distribution with mean vector $\boldsymbol{M}$ and diagonal covariance matrix with diagonal elements $\boldsymbol{\Sigma}$
	\end{algorithmic}
	\label{sec3:algorithm_GPRForPODlatentVariable}
\end{algorithm}
From this subsection, we can obtain the approximated posterior 
\begin{equation}
	q_{\boldsymbol{\phi}}(\boldsymbol{Z}|\mathcal{D},\boldsymbol{\xi};\boldsymbol{\alpha})=q_{\boldsymbol{\phi}}(\boldsymbol{Z}|\boldsymbol{z}_{1:D},\boldsymbol{\xi};\boldsymbol{\alpha})=\Pi_{j=1}^{k}q_{\phi_j}(\boldsymbol{Z}_j|z_{1:D,j},\boldsymbol{\xi},\alpha_j)
	\label{sec3_posterior:CVAE-GPRR}
\end{equation}
where $\boldsymbol{\phi}$ is the set of $\phi_j=[\boldsymbol{\mu}_j,\boldsymbol{\sigma}_j]$ and $\boldsymbol{\alpha}$ is the collection of hyperparameter vectors $\alpha_j=[\sigma_{\pi_{z_j}},l_{j,1},...,l_{j,d},\sigma_{j,GPR}]$ in each GPR model. The form of equation (\ref{sec3_posterior:CVAE-GPRR}) is similar to the mean-field assumption \cite{sec2:CVAE} in CVAE. 
\subsection{Neural Networks for likelihood model}
\label{sec:Neural Networks for likelihood model}
In Section \ref{ssec:Gaussian process regression for POD latent variables}, we project data into a low-dimensional latent space. To reconstruct the states of parametric models, it is natural to use POD basis like equation (\ref{sec3:pod-latent variables}). However, there are two main limitations of this method. First, the accuracy of prediction is highly dependent on the number of POD basis vectors $N_{POD}$. When the true model is unknown and the data is noisy, the choice of $N_{POD}$ is an art. Second, such methods can not give a prediction in the unobserved physical region. Neural networks, good function approximators that have generalization ability, is a useful tool to deal with above issues. We first define the following likelihood model
\begin{equation} U=\mu_U(\boldsymbol{Z},\boldsymbol{x},\boldsymbol{\xi})+\epsilon,\epsilon\mathop{\sim}^{i.i.d}\mathcal{N}\big(0,\sigma_{obs}^2(\boldsymbol{Z},\boldsymbol{x},\boldsymbol{\xi})\big), \ i=1,...,N.
	\label{sec3_eq:LE_proposed}
\end{equation}
There are two parameters in likelihood model, we collect them into a vector
\begin{equation} \theta=[\mu_U(\boldsymbol{Z},\boldsymbol{x},\boldsymbol{\xi}),\sigma_{obs}^2(\boldsymbol{Z},\boldsymbol{x},\boldsymbol{\xi})] \subset\mathbb{R}^2.
	\nonumber
\end{equation}
Then, A fully-connected neural network $LEN_{\Theta}$ with $L-1$ hidden layers and $Relu$ activation functions \cite{sec3:relu} is used to learn $\theta$,\\
\begin{equation}
	\begin{aligned}
		h_{1}&=Relu(W_{1,z}\boldsymbol{Z}+W_{l,\boldsymbol{s}}\boldsymbol{x}+W_{1,\boldsymbol{\xi}}\boldsymbol{\xi}+b_1),\\
		h_{l}&=Relu(W_{l-1}h_{l-1}+b_{l-1}),\ l=2,...,L-1,\\
		\hat{\theta}&=W_{L-1}h_{L-1}+b_{L-1},
	\end{aligned}
     \nonumber
\end{equation}
where $\Theta:=\{W_{1,z},W_{1,\boldsymbol{x}},W_{1,\boldsymbol{\xi}},...,W_L,b_1,b_2,...,b_L\}$ is the set including all weights and bias in the neural network. In this paper, we assume that hidden layers have the same width $N_h$. And the output vector is denoted by a vector $\hat{\theta}:=[LEN_{\Theta}^1(\boldsymbol{Z},\boldsymbol{x},\boldsymbol{\xi}),LEN_{\Theta}^2(\boldsymbol{Z},\boldsymbol{x},\boldsymbol{\xi})]\subset{\mathbb{R}^2}$. To make sure that the estimated noise level $\hat{\sigma}_{obs}^2$ is non-negative. we apply the softplus tranform   $\hat{\sigma}_{obs}^2(\boldsymbol{Z},\boldsymbol{x},\boldsymbol{\xi}):=\log\bigg(1+\exp\big(LEN_{\Theta}^2(\boldsymbol{Z},\boldsymbol{x},\boldsymbol{\xi})\big)\bigg)$ \cite{sec3:softplus}.

\begin{figure}[ht]
	\centering
	\subcaptionbox{CVAE-GPRR\label{sec3_subfig:LE_NN}}
	{\includegraphics[scale=0.32]{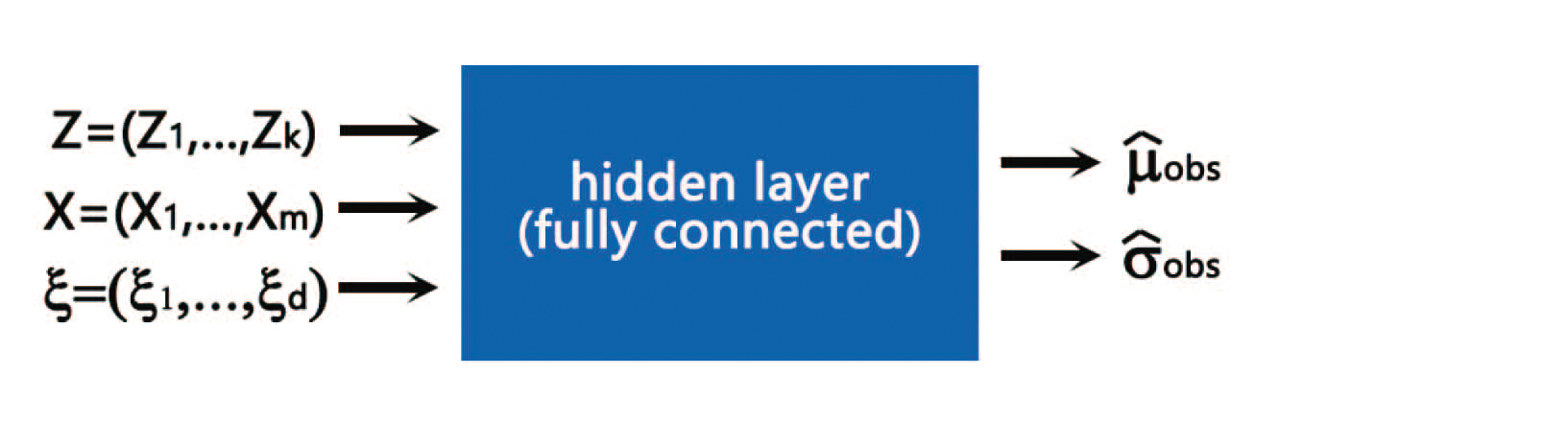}}
	\quad
	\subcaptionbox{CVAE\label{sec3_subfig:LE_CVAE}}
	{\includegraphics[scale=0.32]{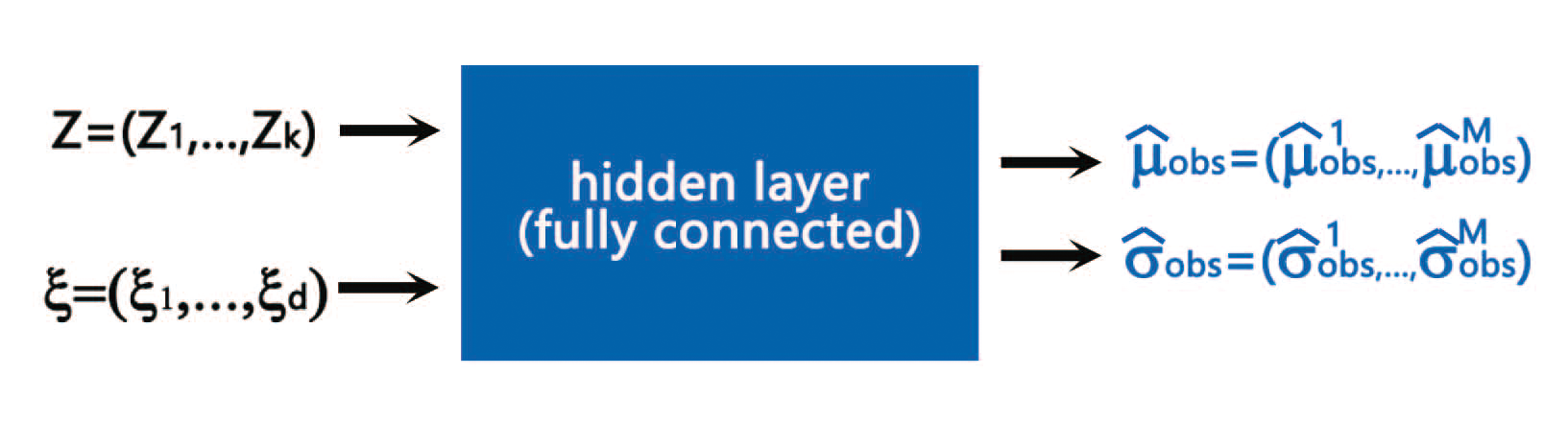}}
\caption{\small{Likelihood models for CVAE-GPRR and CVAE.}}
\label{sec3_fig:LE_NN}
\end{figure}
As shown in Figure \ref{sec3_fig:LE_NN}, the likelihood neural network of CVAE-GPRR has fewer parameters than that of CVAE. For instance, a fully-connected neural network with three hidden layers whose width is $N_h$ has $N_h*(k+m+d+2N_h+5)+2$ parameters in CVAE-GPRR while has $N_h*(k+2M+d+2N_h+3)+2M$ parameters in CVAE. Furthermore, with the same dataset $\boldsymbol{\mathcal{D}}$, the proposed method has $D\times M$ data pairs for training while CVAE only has $D$ training data. We expect better ability to deal with noisy data by using such framework, which will be further explained by numerical results.

\subsubsection{The training of likelihood model}
\label{ssec:the training of likelihood model}
In Section \ref{ssec:Conditional variational autoencoders}, we know that the model evidence can be written as
\begin{equation}
	\begin{aligned}
	\log p_{\theta}(U|\boldsymbol{x},\boldsymbol{\xi})=&E_{{\boldsymbol{\phi}}}\bigg(\log{p_{\theta}(U|\boldsymbol{z},\boldsymbol{x},\boldsymbol{\xi};\Theta)}\bigg)-D_{KL}\big(q_{q_{\boldsymbol{\phi}}}(\boldsymbol{Z}|\mathcal{D},\boldsymbol{\xi};\boldsymbol{\alpha})\Vert p_{\theta}(\boldsymbol{Z}|\boldsymbol{\xi})\big)\\
	&+D_{KL}(q_{\boldsymbol{\phi}}\big(\boldsymbol{Z}|\mathcal{D},\boldsymbol{\xi};\boldsymbol{\alpha})\Vert p_{\theta}(\boldsymbol{Z}|\mathcal{D},\boldsymbol{\xi})\big).
	\label{sec3_eq:model evidence}
	\end{aligned}
\end{equation}
In Section \ref{ssec:GPR recognition model}, we have constructed an approximate posterior $q_{\boldsymbol{\phi}}(\boldsymbol{Z}|\mathcal{D},\boldsymbol{\xi};\boldsymbol{\alpha})$ with GPR recognition model for learning distribution parameter $\boldsymbol{\phi}$. Here, the maximum likelihood method is used to estimate the parameters $\Theta$ in likelihood model, i.e.,
\begin{equation}
	\begin{aligned}
		\Theta^{*}=&\mathop{\arg\max}_{\Theta}\mathbb{E}_{p^{*}(U|\boldsymbol{x},\boldsymbol{\xi})}\bigg(\mathbb{E}_{q_{\boldsymbol{\phi}}}\big(\log p_{\theta}(U|\boldsymbol{z},\boldsymbol{x},\boldsymbol{\xi},\Theta)\big)\bigg)\\
		&\approx \mathop{\arg\min}_{\Theta} \frac{1}{2D\times M}\frac{1}{N_{MC}}\sum_{i=1}^{D\times M}\sum_{j=1}^{N_{MC}}\bigg(\frac{1}{\hat{\sigma}_{obs}^2(\boldsymbol{z}_i^{(j)},\boldsymbol{x}^{(i)},\boldsymbol{\xi}^{(i)})}\big(u(\boldsymbol{x}^{(i)},\boldsymbol{\xi}^{(i)})\\
		&-LEN_{\Theta}^1(\boldsymbol{z}_i^{(j)},\boldsymbol{x}^{(i)},\boldsymbol{\xi}^{(i)})\big)^2+\log \big(2\pi \hat{\sigma}_{obs}(\boldsymbol{z}_i^{(j)},\boldsymbol{x}^{(i)},\boldsymbol{\xi}^{(i)})\big)\bigg),\ \boldsymbol{z}_i^{(j)} \sim q_{\boldsymbol{\phi}}(\boldsymbol{Z}|\mathcal{D},\boldsymbol{\xi};\boldsymbol{\alpha}) .
	\end{aligned}
	\label{sec3_eq:loss fuction}
\end{equation}
The last two terms on the right-hand side of equation (\ref{sec3_eq:model evidence}) are dropped because they are independent of the optimizer. By using gradient-based optimization algorithm such as Adam \cite{sec3:Adam} and backpropagation of neural networks, we can easily train the likelihood model (\ref{sec3_eq:LE_proposed}) by solving the optimization problem (\ref{sec3_eq:loss fuction}).

\subsection{Uncertainty quantification}
\label{ssec:Uncertainty quantification}
Compared to deterministic regression methods, our method can give some uncertainty quantification for the predictive results \cite{sec3:UQ}.

Once the model is trained, the model evidence for any input $\hat{\boldsymbol{x}},\hat{\boldsymbol{\xi}}$ can be approximated by
\begin{equation}
	p_{\theta}(\hat{U}|\boldsymbol{\mathcal{D}},\hat{\boldsymbol{x}},\hat{\boldsymbol{\xi}})\approx\int p_{\theta}(\hat{U}|\hat{\boldsymbol{z}},\hat{\boldsymbol{x}},\hat{\boldsymbol{\xi}})q_{\boldsymbol{\phi}}(\hat{\boldsymbol{Z}}|\boldsymbol{\mathcal{D}},\hat{\boldsymbol{\xi}};\boldsymbol{\alpha})dZ.
	\nonumber
\end{equation}
The predictive mean is
\begin{equation}
	\begin{aligned}
		\mathbb{E}_{p_{\theta}(\hat{U}|\boldsymbol{\mathcal{D}},\hat{\boldsymbol{x}},\hat{\boldsymbol{\xi}})}(\hat{U}|\hat{\boldsymbol{x}},\hat{\boldsymbol{\xi}})
		&\approx \mathbb{E}_{q_{\boldsymbol{\phi}}}\bigg(\mathbb{E}_{p_{\theta}(\hat{U}|\hat{\boldsymbol{z}},\hat{\boldsymbol{x}},\hat{\boldsymbol{\xi}})}(\hat{U}|\hat{\boldsymbol{z}},\hat{\boldsymbol{x}},\hat{\boldsymbol{\xi}})\bigg)\\
		&=\mathbb{E}_{q_{\boldsymbol{\phi}}}\bigg(\mu_u(\hat{\boldsymbol{Z}},\hat{\boldsymbol{x}},\hat{\boldsymbol{\xi}})\bigg)\\
		&\approx \frac{1}{M}\sum_{i=1}^M LEN_{\Theta^{*}}^1(\hat{\boldsymbol{z}}^{(i)},\hat{\boldsymbol{x}},\hat{\boldsymbol{\xi}}), \ \hat{\boldsymbol{z}}^{(i)}\sim q_{\boldsymbol{\phi}}(\hat{\boldsymbol{Z}}|\boldsymbol{\mathcal{D}},\hat{\boldsymbol{\xi}};\boldsymbol{\alpha}).
	\end{aligned}
	\label{sec3_eq:mymethod_mean}
\end{equation}
The predictive variance can be obtained by the law of total variance
\begin{equation}
	\begin{aligned}
		&Var_{p_{\theta}(\hat{U}|\boldsymbol{\mathcal{D}},\hat{\boldsymbol{x}},\hat{\boldsymbol{\xi}})}(\hat{U}|\hat{\boldsymbol{x}},\hat{\boldsymbol{\xi}})\\
		&\approx \mathbb{E}_{q_{\boldsymbol{\phi}}}\bigg(Var_{p_{\theta}(\hat{U}|\hat{\boldsymbol{z}},\hat{\boldsymbol{x}},\hat{\boldsymbol{\xi}})}(\hat{U}|\hat{\boldsymbol{z}},\hat{\boldsymbol{x}},\hat{\boldsymbol{\xi}})\bigg)+Var_{q_{\boldsymbol{\phi}}}\bigg(\mathbb{E}_{p_{\theta}(\hat{U}|\hat{\boldsymbol{z}},\hat{\boldsymbol{x}},\hat{\boldsymbol{\xi}})}(\hat{U}|\hat{\boldsymbol{z}},\hat{\boldsymbol{x}},\hat{\boldsymbol{\xi}})\bigg)\\
		&=\mathbb{E}_{q_{\boldsymbol{\phi}}}\bigg(\sigma_{obs}^2(\hat{\boldsymbol{Z}},\hat{\boldsymbol{x}},\hat{\boldsymbol{\xi}})\bigg)+Var_{q_{\boldsymbol{\phi}}}\bigg(\mu_u(\hat{\boldsymbol{Z}},\hat{\boldsymbol{x}},\hat{\boldsymbol{\xi}})\bigg)\\
		&=\mathbb{E}_{q_{\boldsymbol{\phi}}}\bigg(\sigma_{obs}^2(\hat{\boldsymbol{Z}},\hat{\boldsymbol{x}},\hat{\boldsymbol{\xi}})+\mu_u^2(\hat{\boldsymbol{Z}},\hat{\boldsymbol{x}},\hat{\boldsymbol{\xi}})\bigg)-\mathbb{E}_{q_{\boldsymbol{\phi}}}^2\bigg(\mu_u(\hat{\boldsymbol{Z}},\hat{\boldsymbol{x}},\hat{\boldsymbol{\xi}})\bigg)\\
		&\approx \frac{1}{M}\sum_{i=1}^M \hat{\sigma}_{obs}^{2*}(\hat{\boldsymbol{z}}^{(i)},\hat{\boldsymbol{x}},\hat{\boldsymbol{\xi}})+\bigg(LEN_{\Theta^{*}}^1(\hat{\boldsymbol{z}}^{(i)},\hat{\boldsymbol{x}},\hat{\boldsymbol{\xi}})\bigg)^2-\bigg(\frac{1}{M}\sum_{i=1}^M LEN_{\Theta^{*}}^1(\hat{\boldsymbol{z}}^{(i)},\hat{\boldsymbol{x}},\hat{\boldsymbol{\xi}})\bigg)^2,
	\end{aligned}
	\label{sec3_eq:mymethod_var}
\end{equation}
where $\hat{\boldsymbol{z}}^{(i)}\sim q_{\boldsymbol{\phi}}(\hat{\boldsymbol{Z}}|\boldsymbol{\mathcal{D}},\hat{\boldsymbol{\xi}};\boldsymbol{\alpha})$ and $\hat{\sigma}_{obs}^{2*}(\boldsymbol{Z},\boldsymbol{x},\boldsymbol{\xi}) =\log\bigg(1+\exp\big(LEN_{\Theta^{*}}^2(\boldsymbol{Z},\boldsymbol{x},\boldsymbol{\xi})\big)\bigg)$.

\subsection{Comparison with Conditional Variational Autoencoder}
\label{ssec:Comparison with Conditional Variational Autoencoder}
\begin{figure}[ht]
	\centering
	\subcaptionbox{CVAE-GPRR\label{sec3_subfig:framework_Proposed}}
	{\includegraphics[scale=0.45]{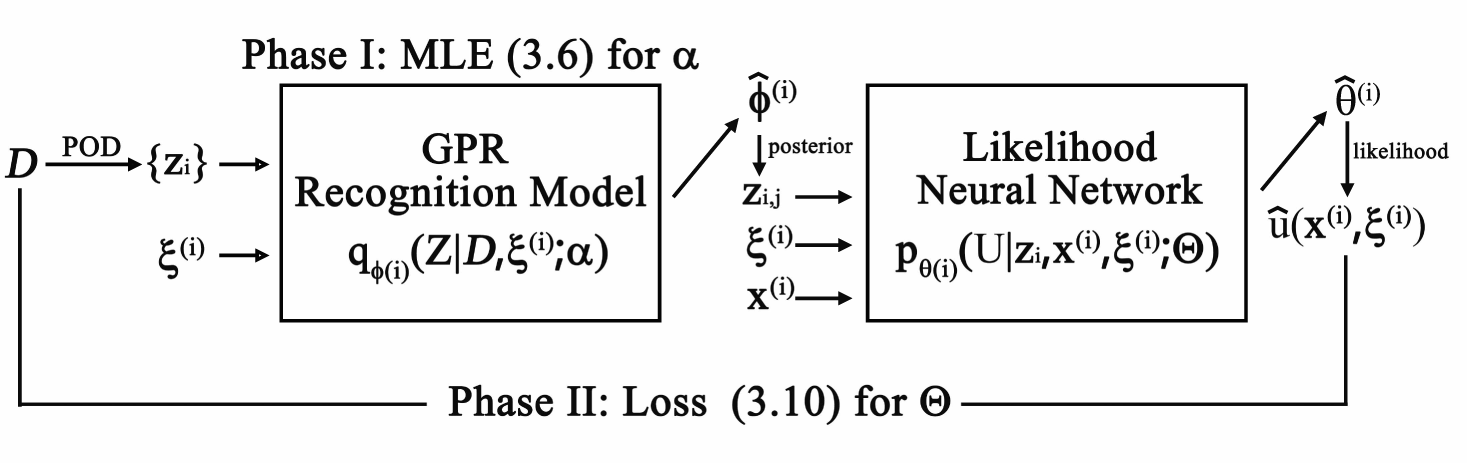}}
	\quad
	\subcaptionbox{CVAE\label{sec3_subfig:framework_CVAE}}
	{\includegraphics[scale=0.45]{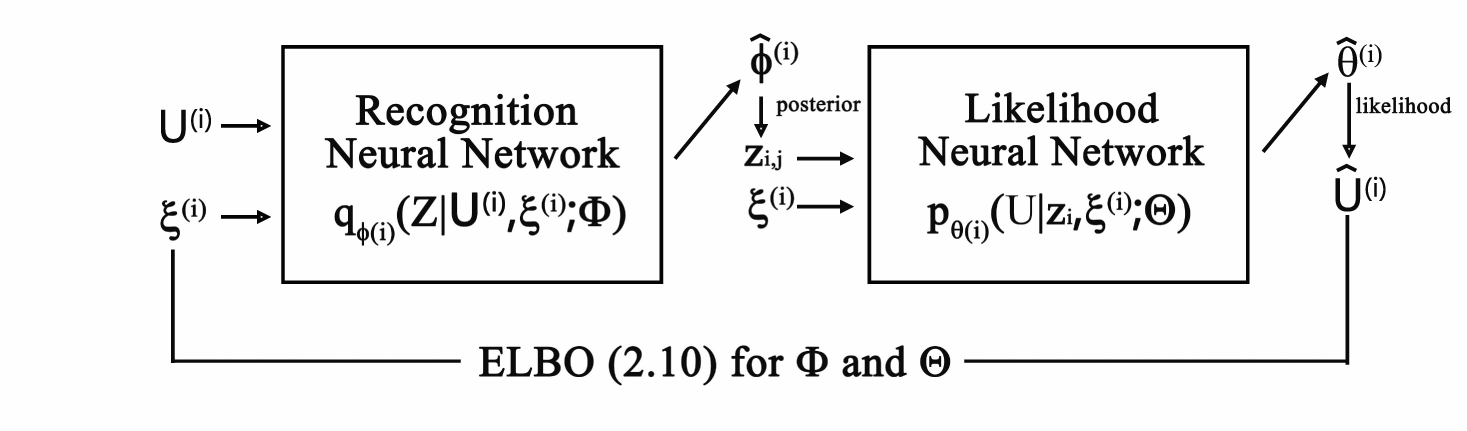}}
	\caption{Schemas of CVAE-GPRR and CVAE.}
	\label{fig_framework}
\end{figure}
In the last section, we will discuss the connections and differences between CVAE-GPRR and CVAE.
Figure \ref{fig_framework} presents the schemas of CVAE-GPRR and CVAE. We note that both methods consist of two models: recognition model and likelihood model. The aim of the recognition model is to filter redundant information from data and denoise them at the same time. CVAE-GPRR achieves this goal in a more interpretable way. We first use POD to obtain the principle components and then project data into latent space with POD basis. Therefore, we can get observations of POD latent variables $z$ although we call them as latent variables, which is different from that of CVAE. Next, GPR is used to alleviate the influence of noise and learn the mapping from parameters to POD latent variables (see Figure \ref{sec3_fig:GPR_effect}). However, CVAE packed all these processes in a black box — recognition neural network. In addition, our recognition model has fewer parameters by utilizing non-parametric method GPR, which leads to lower optimization complexity. The likelihood model is used to reconstruct data from latent space. Compared with CVAE, our framework can make prediction in the unobserved physical region by adding physical variables into the inputs of likelihood neural networks. However, CVAE needs retraining or interpolation to achieve this goal.
\begin{figure}[htbp]
	\setlength{\abovecaptionskip}{0.cm}		
	\centering		
	\includegraphics[scale=0.4]{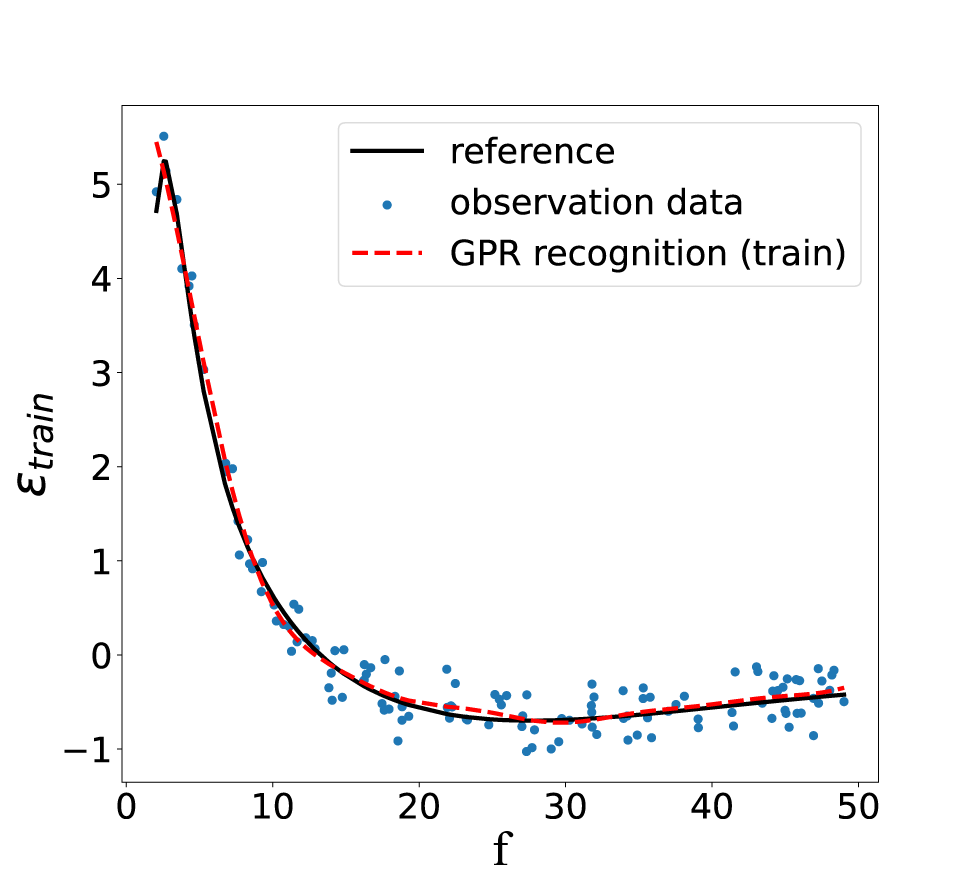}		
	\caption{\small{GPR can alleviate the influence of noise (1D real Morlet wavelet function, see Section\ref{sec4_eq:morlet}): GPR for the first POD latent variable with $n=3$. The reference and observations are obtained by projecting the truth and observation data into the feature space spanned by POD basis of the snapshot $\boldsymbol{\mathcal{S}}$.}}	
	\label{sec3_fig:GPR_effect}	
\end{figure}
For test, CVAE use the prior of latent variables to generate new samples of data, which is not consistent with training process that pass the samples drawn from posterior through the likelihood model. Sometimes, this inconsistency will lead to worse generalization performance \cite{intro:CVAE}. As shown in Figure \ref{fig_framework}, CVAE-GPRR can first approximate posterior using GPR recognition and we can thus avoid above inconsistency for CVAE. From this point of view, we can understand our loss function in another way. Simply let prior equal to posterior in ELBO (\ref{sec2:CVAE_ELBO}), we can also arrive at optimization problem (\ref{sec3_eq:loss fuction}). And then the results of GPR recognition model can be seen as an informative prior of latent variable in DLVM (\ref{sec2_eq:LVM}).
\section{Numerical results}
\label{sec:Numerical results and discussions}
In this section, the numerical results for various parametric  models  are presented to show the accuracy and efficiency of CVAE-GPRR, which includes the generalization ability with respect to  parameters and  physical variables and the results of uncertainty quantification. In each numerical example, training data are generated by solving  the given models, which are assumed known for having a reference truth. The observation data is equal to the truth with  Gaussian noise perturbation.

In Section \ref{ssec:1D real morlet wavelet function}, CVAE-GPRR is used to a simple 1D real morlet wavelet function with two shape parameters.  Based on the discussions in Section \ref{ssec:Comparison with CVAE}, we further compare the proposed method with CVAE numerically. To show the advantages of CVAE-GPRR over GPR-based ROM, we repeat the experiments given different numbers of POD latent variables for these two methods. In Section \ref{ssec:2D parametric diffusion problem}, we show the numerical results of a diffusion problems with three parameters in diffusion coefficients and compare the predictive results for different priors of POD latent variables. In Section \ref{ssec:2D parametric p-Laplacian equation}, we collect data from p-Laplacian equations, which are nonlinear equations, with physical parameter $p$ and additive three parameters in the forcing term. In this test case, we also show the influence of the number of training data. In Section \ref{ssec:2D skewed lid-driven cavity in parametric region}, we use the data simulated from a skewed lid-driven cavity problem to train CVAE-GPRR. Compared to above three test cases, the relationship between parameters and state variables is not very direct since the parameters are set in the physical region.
\subsection{1D real morlet wavelet function}
\label{ssec:1D real morlet wavelet function}
Consider the 1D real morlet wavelet function, which is defined as the product of a cosine wave and a Gaussian window:
\begin{equation}
	u(x,\boldsymbol{\xi})=\cos(2\pi fx)\exp(\frac{-x^2}{2h^2}),
	\label{sec4_eq:morlet}
\end{equation}
where $x$ is time in seconds. To avoid introducing a phase shift, the region of $x$ is centered at the coordinate origin. Here, we limit $x$ from -1 to 1 seconds. The morlet wavelets (\ref{sec4_eq:morlet}) are controlled by two parameters: frequency $f(Hz)$ and the width of the Gaussian $h$. In practical application, the second parameter is further defined as
\begin{equation}
	h=\frac{n}{2\pi f},
	\nonumber
\end{equation}
where the integer n is often called the “number of cycles” that determine the time-frequency precision trade-off \cite{sec4:morletfunction}. In the following experiment, we take the value of $n$ ranging from 2 to 5 over frequencies between 2 Hz and 50 Hz, which is within the typical setting for neurophysiology data such as EEG, MEG and LFP \cite{sec4:morletfunction}. Now, we introduce the parameter vector
\begin{equation}
	\boldsymbol{\xi}=[f,n]\in \Xi:=[2,50]\times\{2,3,4,5\}\subset \mathbb{R}\times\mathbb{N},
	\label{sec4_eq:morlet_parameter_space}
\end{equation}

Examples of functions for two different parameters are shown in Figure \ref{sec4_fig:morlet_example}.

\begin{figure}[htbp]	
	\centering		
	\includegraphics[scale=0.2]{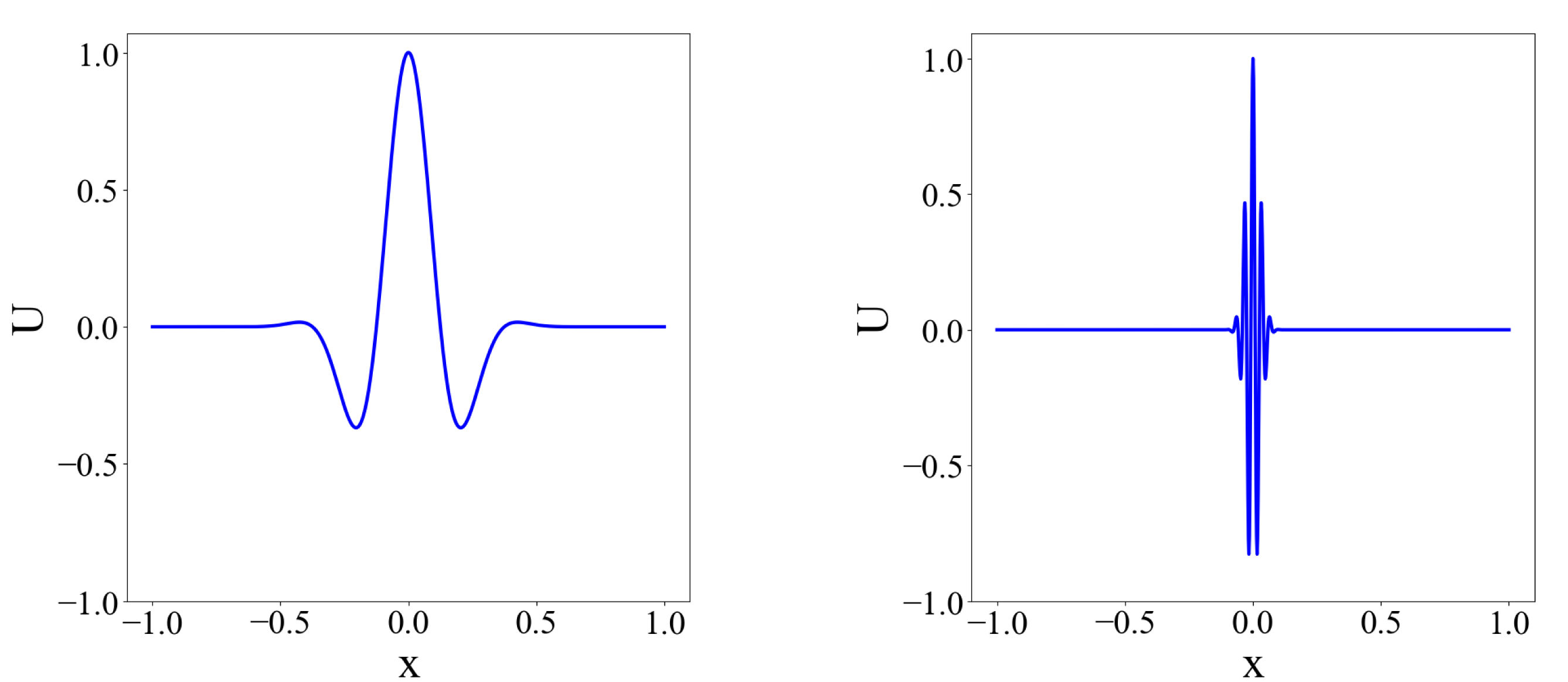}		
	\caption{\small{Examples of functions defined as equation (\ref{sec4_eq:morlet}) for the parameters  $\boldsymbol{\xi}^{(1)}=(2.,2)$ (left) and $\boldsymbol{\xi}^{(2)}=(30.,5)$} (right).}	
	\label{sec4_fig:morlet_example}	
\end{figure}

In the numerical experiment we divide the time region $\Omega:=[-1,1]$ into 500 equal intervals and the parameters are sampled from uniform distribution in $\Xi$. We collect  snapshots of $u$ for 1000 different parameter values, half of which are selected as training data and the rest are used to test the generalization ability of CVAE-GPRR. The training set is also used to learn GPR recognition model besides being used for training the likelihood model. To simulate the realistic signal, the white Gaussian noise is added to data as sensor noise \cite{sec4:eeg_noise}.

In this test case, ten POD latent variables are chosen by performing Algorithm \ref{sec3:algorithm_GPRForPODlatentVariable}. The hyperparameters for training likelihood model are given in the first row of Table \ref{sec4_exper1:morlet_table_hyper}. During training, we gradually reduce the learning rate to improve the training accuracy. Specifically, we first train the model with 0.001 learning rate for 100 epoches and then divide the learning rate by 10 every 50 epoches.

\begin{table}[htbp]
	\centering  
	\setlength{\abovecaptionskip}{0cm}
	\setlength{\belowcaptionskip}{0.2cm}
	\caption{Hyperparameters of likelihood model for 1D real morlet wavelet function.}  
	\label{sec4_exper1:morlet_table_hyper}  
	\scalebox{0.8}{
		\begin{tabular}{cccccc}
			\toprule  
			Model & Layer structure & Optimizer & Learning rate & Epoch & Batch size \\  
			\midrule
			& & & & & \\[-6pt]  
			CVAE-GPRR &$N_{POD}$+3-100-100-100-100-2 &\multirow{2}*{Adam} & \multirow{2}*{0.001-0.0001-0.00001} & \multirow{2}*{100-50-50} & \multirow{2}*{1000} \\
			\cline{1-2}
			& & & & & \\[-6pt]  
			Discrete likelihood model &12-100-100-100-100-1002 & & & & \\
			\bottomrule
		\end{tabular}
	}
\end{table}

\subsubsection{Comparison with conditional variational autoencoder}
\label{ssec:Comparison with CVAE}
In Section \ref{ssec:Comparison with Conditional Variational Autoencoder}, we have compared CVAE-GPRR with CVAE in details. In this subsection, we will further illustrate the influence of different designs of recognition model and likelihood model by numerical experiments.\\
\hspace*{\fill} \\
\noindent\textbf{The efficacy of likelihood model.}
We first keep recognition model the same and compare different framework of likelihood model. As shown in Figure \ref{sec3_subfig:LE_NN}, the physical variables $\boldsymbol{x}$ (time $x$ in this subsection have the same role) are added into the inputs. Thus, the outputs are continuous with respect to $\boldsymbol{x}$. However, some existing frameworks \cite{sec4:framework1,sec4:framework2} including CVAE design the likelihood model as a reconstruction of data, namely, the outputs approximate the values of the system at some discrete points in the physical region $\Omega$ (see Figure \ref{sec3_subfig:LE_CVAE}). To distinguish this framework from the proposed one, we call it discrete likelihood model and its hyperparameters for training are demonstrated in the second row of Table \ref{sec4_exper1:morlet_table_hyper}. For comparison, Figure \ref{sec4_exper1:fig_framework} displays the predictive mean of these two frameworks at one test instance of the parameters. As can be seen from the figure, CVAE-GPRR can better recognize the original signal from noisy data, which is more obvious as the noise level increases.
\begin{figure}[htbp]		
	\centering		
	\includegraphics[scale=0.15]{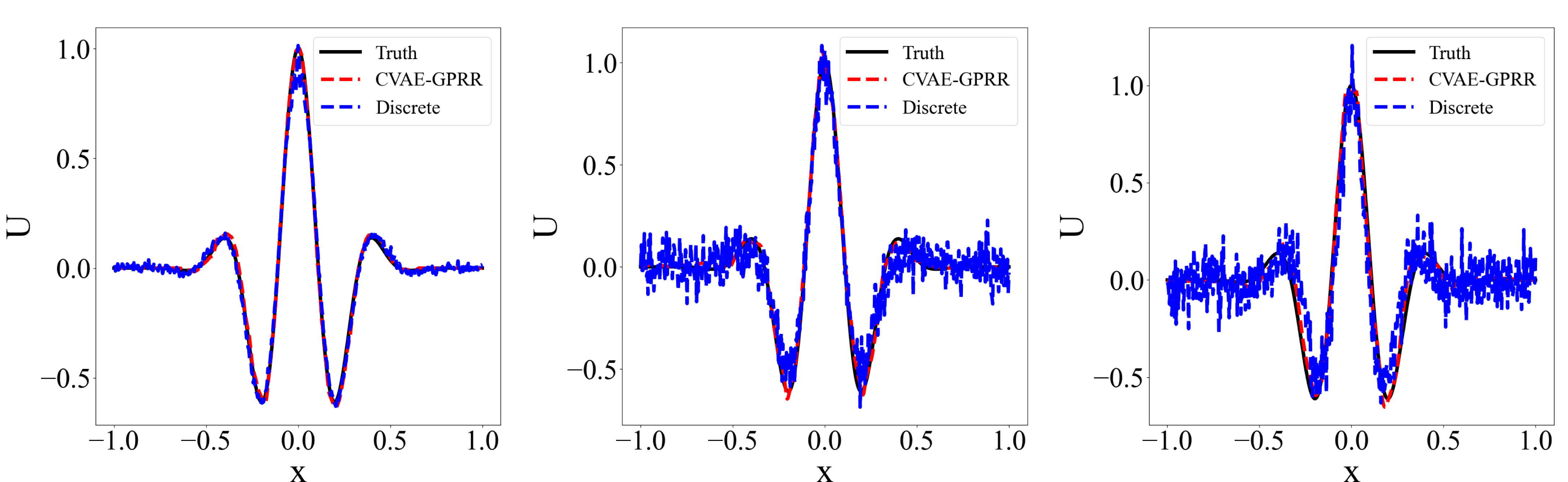}		
	\caption{\small{1D real morlet test case: the predictive mean of CVAE-GPRR and discrete likelihood model at the test instance of the parameters $\boldsymbol{\xi}=(2.28,3)$. From left to right, the noise levels are 0.01,0.1,0.2, respectively.}}	
	\label{sec4_exper1:fig_framework}	
\end{figure}
Furthermore, to illustrate numerical results quantitatively, we define the relative test mean error $\epsilon_{test}$ as follows, which is also used in the other experiments.
\begin{equation}
	\epsilon_{test}=\frac{1}{N_{test}}\sum_{i=1}^{N_{test}}\frac{\Vert\boldsymbol{\mu}_{NN}(\boldsymbol{x}_{1:M},\boldsymbol{\xi}^{(i)})-\boldsymbol{U}^{(i)}\Vert_2}{\Vert\boldsymbol{U}^{(i)}\Vert_2},
	\label{sec4_eq:relative test mean error}
\end{equation}
where $N_{test}$ is the size of test set and $\Vert\cdot\Vert_2$ is vector 2-norm. The relative test mean errors of two frameworks are reported in Table \ref{sec4_exper1:table_error_framework}, which illustrates that the prediction ability of CVAE-GPRR with respect to parameters has an significant advantage over that of discrete likelihood model especially when the noise level is large. As mentioned in Section \ref{ssec:Comparison with Conditional Variational Autoencoder}, the discrete likelihood model of CVAE has more neural network parameters but fewer training data pairs than those of CVAE-GPRR if the width of neural networks are the same, which may lead to overfitting, i.e., the model wrongly learns the noise in training data (as shown in Figure \ref{sec4_exper1:fig_framework}). Extra regularization tricks are needed. In addition, our method can also predict the values of morlet functions in the unobserved time region without retraining the model while the discrete likelihood model can only estimate function values at fixed points in $\Omega$. We partition the time region by 1000 equal intervals (It is called fine grid and the original partition is called coarse grid.) and compute the corresponding relative test mean error that is recorded in the fourth column of Table \ref{sec4_exper1:table_error_framework}.

\begin{table}[htbp]
	\centering  
	\setlength{\abovecaptionskip}{0cm}
	\setlength{\belowcaptionskip}{0.2cm}
	\caption{1D real morlet test case: comparison of relative mean errors for CVAE-GPRR and discrete likelihood model.}  
	\label{sec4_exper1:table_error_framework}  
	\begin{threeparttable}
		\scalebox{0.8}{
			\begin{tabular}{c|ccc}
				\toprule  
				\diagbox[width=4cm,height=1.5cm]{Noise level}{$\epsilon_{test}$}{Model} & \tabincell{c}{CVAE-GPRR\\(coarse grid)} &\tabincell{c}{Discrete likelihood model\\(coarse grid)} & \tabincell{c}{CVAE-GPRR\\(fine grid)}\\   
				\hline
				& & & \\[-6pt]  
				0.01 &$3.0504\times 10^{-2}$ &$3.6333\times 10^{-2}$ & $3.8213\times 10^{-2}$ \\
				& & & \\[-6pt]  
				0.1 & $5.8427\times 10^{-2}$ &$3.0955\times 10^{-1}$ & $6.1505\times 10^{-2}$ \\
				& & & \\[-6pt]  
				0.2 & $7.9915\times 10^{-2}$ &$4.4050\times 10^{-1}$ & $8.1501\times 10^{-2}$ \\
				\bottomrule
		\end{tabular}}
	\end{threeparttable}       
\end{table}

\hspace*{\fill} \\
\noindent\textbf{The efficacy of Gaussian process regression recognition.}
In this subsection, the layer structure of likelihood model for CVAE are designed to be the same as that of CVAE-GPRR (see Table \ref{sec4_exper1:morlet_table_hyper}) in the following experiment. And the layer structure of the recognition model and the likelihood model in CVAE are designed symmetrical. We compare the performance of GPR recognition model in CVAE-GPRR with neural network recognition model in CVAE. The relative test mean errors are reported in Table \ref{sec4_exper1:table_error_approach}. However,in this test case, the proposed method has similar performance to that of CVAE while the training complexity of CVAE-GPRR is lower than that of CVAE due to fewer model parameters and parallel GPR of POD latent variables (see analysis in Section \ref{sec:Neural Networks for likelihood model}). 
In addition, we me

\subsubsection{Comparison with GPR-based reduced order modeling method}
 \noindent In this section, we compare the predictive performance with a GPR-based ROM \cite{intro:GPR-based RB}, which reconstructs observations by directly using POD basis. Since the accuracy of GPR-based ROM are dependent on the number of POD latent variables $N_{POD}$, we let $N_{POD}$ equal $1,2,3,4,5,10,20,30,40$, respectively and record the minimal $\epsilon_{test}$. From Table \ref{sec4_exper1:table_error_approach}, it can be found that GPR-based ROM can not work well when the noise level is high.\par
\begin{table}[htbp]
	\centering  
	\setlength{\abovecaptionskip}{0cm}
	\setlength{\belowcaptionskip}{0.2cm}
	\caption{1D real morlet test case: comparison of relative test mean errors for different approaches on the coarse grid.}  
	\label{sec4_exper1:table_error_approach}  
	\begin{threeparttable}
		\scalebox{0.8}{
			\begin{tabular}{c|ccc}
				\toprule  
				\diagbox[width=4cm,height=1.5cm]{Noise level}{$\epsilon_{test}$}{Method} & \tabincell{c}{CVAE-GPRR\\($N_{POD}=10$)} &CVAE & GPR-based ROM\\   
				\hline
				& & & \\[-6pt]  
				0.01 &$3.0504\times 10^{-2}$ &$2.3852\times 10^{-2}$ & $2.8143\times 10^{-2}(N_{POD}=20)$ \\
				& & & \\[-6pt]  
				0.1 & $5.8427\times 10^{-2}$ &$5.2660\times 10^{-2}$ & $1.6900\times 10^{-1}(N_{POD}=20)$ \\
				& & & \\[-6pt]  
				0.2 & $7.9915\times 10^{-2}$ &$8.6693\times 10^{-2}$ & $2.9494\times 10^{-1}(N_{POD}=10)$ \\
				\bottomrule
		\end{tabular}}
	\end{threeparttable}       
\end{table}
To further explain the influence of the number of POD latent variables of our method, we train the likelihood model with $N_{POD}=1,2,3,4,5,10,20,30$, respectively. It is shown in Table \ref{sec4_exper1:table_n_pod} that the accuracy of CVAE-GPRR is not sensitive to the number of POD latent variables since the proposed method does not use POD basis to reconstruct data, which is necessary for GPR-based ROM. We thus do not need to choose $N_{POD}$ carefully and the proposed method is more practical when the truth is not available.
\begin{table}[htbp]
	\centering  
	\setlength{\abovecaptionskip}{0cm}
	\setlength{\belowcaptionskip}{0.2cm}
	\caption{1D real morlet test case (noise level: 0.01): comparison of relative test mean errors for CVAE-GPRR and GPR-based ROM.}
	\label{sec4_exper1:table_n_pod}  
	\scalebox{0.6}{
		\begin{tabular}{c|cccccccc}
			\toprule  
			\diagbox[width=4cm,height=1.5cm]{Method}{$\epsilon_{test}$}{$N_{POD}$} & 1  & 2 & 3 & 4 & 5 & 10 & 20 & 30\\   
			\hline
			& & & & & & & &\\[-6pt]  
			CVAE-GPRR &$3.1163\times 10^{-2}$ &$2.5336\times 10^{-2}$ &$3.4603\times 10^{-2}$&$2.6987\times 10^{-2}$ &$3.5202\times 10^{-2}$ & $3.0504\times 10^{-2}$ &$3.4890\times 10^{-2}$ &$2.9820\times 10^{-2}$\\
			& & & & & & & &\\[-6pt]  
			GPR-based ROM&$6.5863\times 10^{-1}$ &$4.9926\times 10^{-1}$ &$3.8576\times 10^{-1}$&$3.0430\times 10^{-1}$ &$2.2433\times 10^{-1}$ & $8.0198\times 10^{-2}$ &$2.8143\times 10^{-2}$ &$2.8276\times 10^{-2}$\\
			\bottomrule
	\end{tabular}}
\end{table}

\subsection{Parametric diffusion problem}
\label{ssec:2D parametric diffusion problem}
In the second test case, we consider a parametric system described by a second-order diffusion problem \cite{sec4:framework1} of the form:
\begin{equation}
	\left\{
	\begin{aligned}
		&-div\big(\kappa(\boldsymbol{x},\boldsymbol{\xi})\nabla u(\boldsymbol{x},\boldsymbol{\xi})\big)=1 \ &in \ \Omega,\\
		&u(\boldsymbol{x},\boldsymbol{\xi})=1 \ &on \  \partial{\Omega}_D^1, \\
		&u(\boldsymbol{x},\boldsymbol{\xi})=0 \ &on \ \partial{\Omega}_D^0,\\
		&\frac{\partial{u}}{\partial{\vec{n}}}=0 \ &on \ \partial{\Omega}_N,\\	
	\end{aligned}
	\right.
	\label{sec4_exper2:equation_definition}
\end{equation}
where physical region $\Omega=(0,1)^2\subset \mathbb{R}^2$ describes a square beam and $u$ is the temperature of the beam. The following boundaries are defined
\begin{equation}
	\partial{\Omega}_N=\{\boldsymbol{x}:=(x_1,x_2)\in\partial{\Omega}:x_2=1\}, \ \partial{\Omega}_D^1=\{\boldsymbol{x}\in\partial{\Omega}:x_2=0\},\ \partial{\Omega}_D^0=\partial{\Omega}\backslash\partial{\Omega}_N\backslash \partial{\Omega}_D^1,
	\nonumber
\end{equation}
that consist of the whole boundary, i.e., $\partial{\Omega}=\partial{\Omega}_N\bigcup \partial{\Omega}_D^1 \bigcup \partial{\Omega}_D^0$. The diffusion coefficient is parametrized as follows
\begin{equation}
	\kappa(\boldsymbol{x},\boldsymbol{\xi})=\xi_3+\frac{1}{\xi_3}exp(-\frac{\sum_{i=1}^2(x_i-\xi_i)^2}{\xi_3}),\ \boldsymbol{\xi}=(\xi_1,\xi_2,\xi_3),
	\nonumber
\end{equation}
where the parameter vector $\boldsymbol{\xi}$ uniformly takes values on\par
\begin{equation}
	\Xi:=[0.4,0.6]^2\times[0.05,0.1]\subset\mathbb{R}^3.
\end{equation}
\indent We generate a uniform mesh of physical region $\Omega$ with $101 \times 101$ nodes (fine mesh) and finite element method (FEM) with bilinear polynomials is used to obtain reference solutions for 1000 different parameters, which are divided into training and test sets with an 600-400 split. Examples of solutions for three different parameters are displayed in Figure \ref{sec4_exper2:fig_equation_eg}. Since CVAE-GPRR has generalization ability with respect to physical variables $\boldsymbol{x}$, we only take values of solutions on $51\times 51$ nodes (coarse mesh) of the fine mesh for training. Besides, the FEM solutions are disturbed by a white Gaussian noise and twenty POD latent variables are chosen.

\begin{figure}[ht]		
	\centering		
	\includegraphics[scale=0.30]{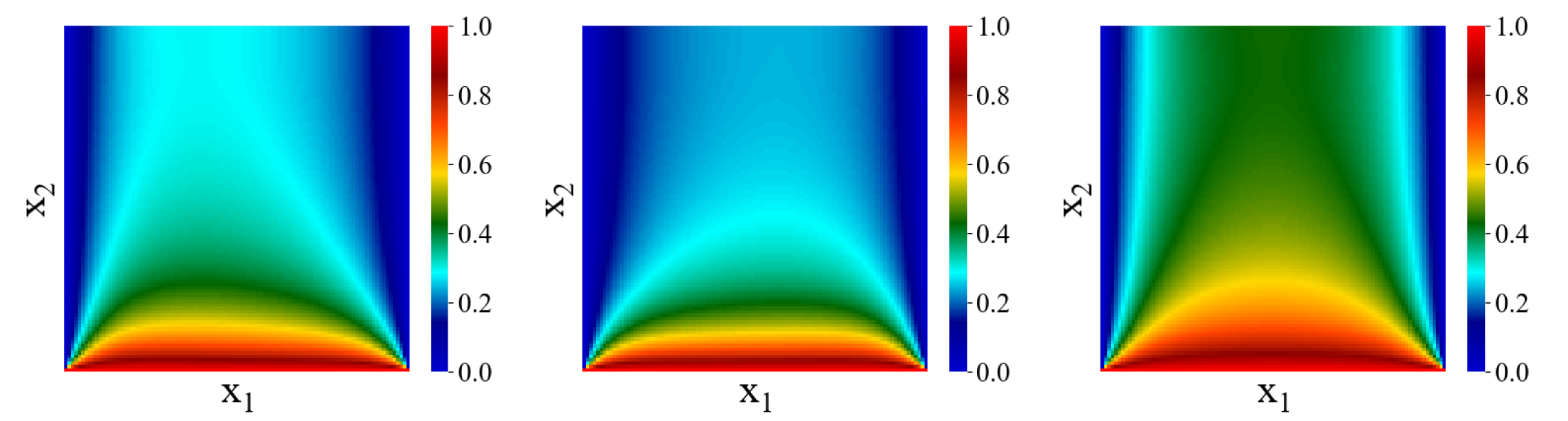}		
	\caption{\small{Examples of solutions (fine mesh) of the diffusion problem (\ref{sec4_exper2:equation_definition}) for the parameters  $\boldsymbol{\xi}^{(1)}=(0.44,0.54,0.10)$(left), $\boldsymbol{\xi}^{(2)}=(0.53,0.43,0.07)$}(middle) and  $\boldsymbol{\xi}^{(3)}=(0.46,0.60,0.06)$(right).}	
	\label{sec4_exper2:fig_equation_eg}	
\end{figure}

In this test case, the likelihood model are trained by the following hyperparameters in Table \ref{sec4_exper2:table_hyper}.

\begin{table}[ht]
	\centering  
	\setlength{\abovecaptionskip}{0cm}
	\setlength{\belowcaptionskip}{0.2cm}
	\caption{Hyperparameters of likelihood model for the parametric diffusion problem.}  
	\label{sec4_exper2:table_hyper}  
	\scalebox{0.8}{
		\begin{tabular}{ccccc}
			\toprule  
			Layer structure & Optimizer & Learning rate & Epoch & Batch size \\  
			\midrule
			& & & &  \\[-6pt]  
			25-200-200-200-200-2&Adam&0.0001&100&1000 \\
			\bottomrule
		\end{tabular}
	}
\end{table}

For a new parameter vector, we can estimate the mean (\ref{sec3_eq:mymethod_mean}) of the corresponding solution straightforwardly by calculating GPR prediction and taking them and physical variables into the trained likelihood model, which avoids solving a large algebra equation that is necessary for FEM. For 400 unobserved test parameter vectors, we estimate the means of their solutions with 500 POD latent variables samples and solve the corresponding equations by FEM on NVIDIA TESLA P100 GPUs. The average computational time of CVAE-GPRR over test set is $\textbf{0.396s}$, which is one hundred times faster than that of FEM ($\textbf{101.628s}$). Figure \ref{sec4_exper2:fig_equation_result} illustrates that CVAE-GPRR can predict the solution of equation for unobserved parameters and physical variables with a tolerable accuracy loss even though the training data is noisy. The predictive variance (\ref{sec3_eq:mymethod_var}) of CVAE-GPRR is also displayed in the right-most column Figure \ref{sec4_exper2:fig_equation_result}, which increases as noise level increases and is close to observation noise.

\begin{figure}[htbp]		
	\centering		
	\includegraphics[scale=0.21]{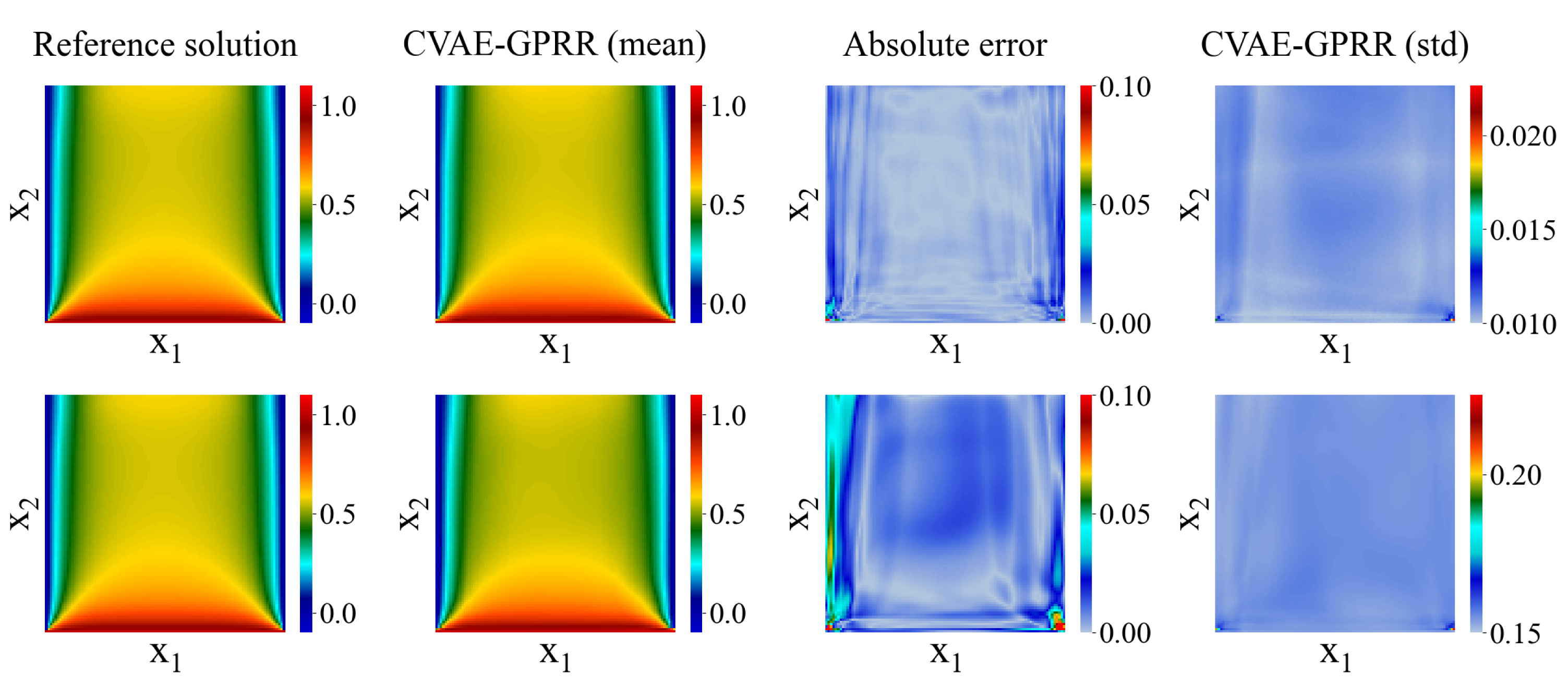}		
	\caption{\small{The parametric diffusion problem case: the predictive result (fine mesh) of CVAE-GPRR at the test instance of the parameters $\boldsymbol{\xi}=(0.50, 0.45, 0.06)$. From top to bottom, the noise level is 0.01,0.15, respectively.}}	
	\label{sec4_exper2:fig_equation_result}	
\end{figure}

In Section \ref{ssec:Comparison with Conditional Variational Autoencoder}, we note that the results of GPR recognition model can also be seen as a prior of latent variables. At the end of this subsection,  we replace the GPR prior with a trival setting — standard Gaussian distribution, i.e., $\mathcal{N}(0,1)$ and retrain the likelihood model using the same hyperparameters in Table \ref{sec4_exper2:table_hyper}. It can be seen from Table \ref{sec4_exper2:table_rel_error} that the performance of CVAE-GPRR is a bit better than that of $\mathcal{N}(0,1)$ in this test case. However, the loss function of CVAE-GPRR declines faster in Figure \ref{sec4_exper2:fig_prior_loss}, which illustrates that the proposed method help speeding up the training process.
\begin{table}[htbp]
	\centering  
	\setlength{\abovecaptionskip}{0cm}
	\setlength{\belowcaptionskip}{0.2cm}
	\caption{The parametric diffusion problem case: comparison of relative test mean errors for different prior.}  
	\label{sec4_exper2:table_rel_error}  
	\begin{threeparttable}
		\scalebox{0.8}{
			\begin{tabular}{c|cccc}
				\toprule  
				\diagbox[width=4cm,height=1.5cm]{Noise level}{$\epsilon_{test}$}{Method} & \tabincell{c}{CVAE-GPRR\\(coarse mesh)} &\tabincell{c}{$\mathcal{N}(0,1)$\\(coarse mesh)} & \tabincell{c}{CVAE-GPRR\\(fine mesh)}&\tabincell{c}{$\mathcal{N}(0,1)$\\(fine mesh)}\\   
				\hline
				& & & & \\[-6pt]  
				0.01 &$9.5012\times 10^{-3}$ &$1.2345\times 10^{-2}$ & $1.9056\times 10^{-2}$ &$2.4658\times 10^{-2}$\\
				& & & &\\[-6pt]  
				0.15 &$2.5350\times 10^{-2}$ &$3.9123\times 10^{-2}$ & $3.4258\times 10^{-2}$ &$4.6054\times 10^{-2}$\\
				\bottomrule
		\end{tabular}}
	\end{threeparttable}       
\end{table}
\begin{figure}[htbp]		
	\centering		
	\includegraphics[scale=0.25]{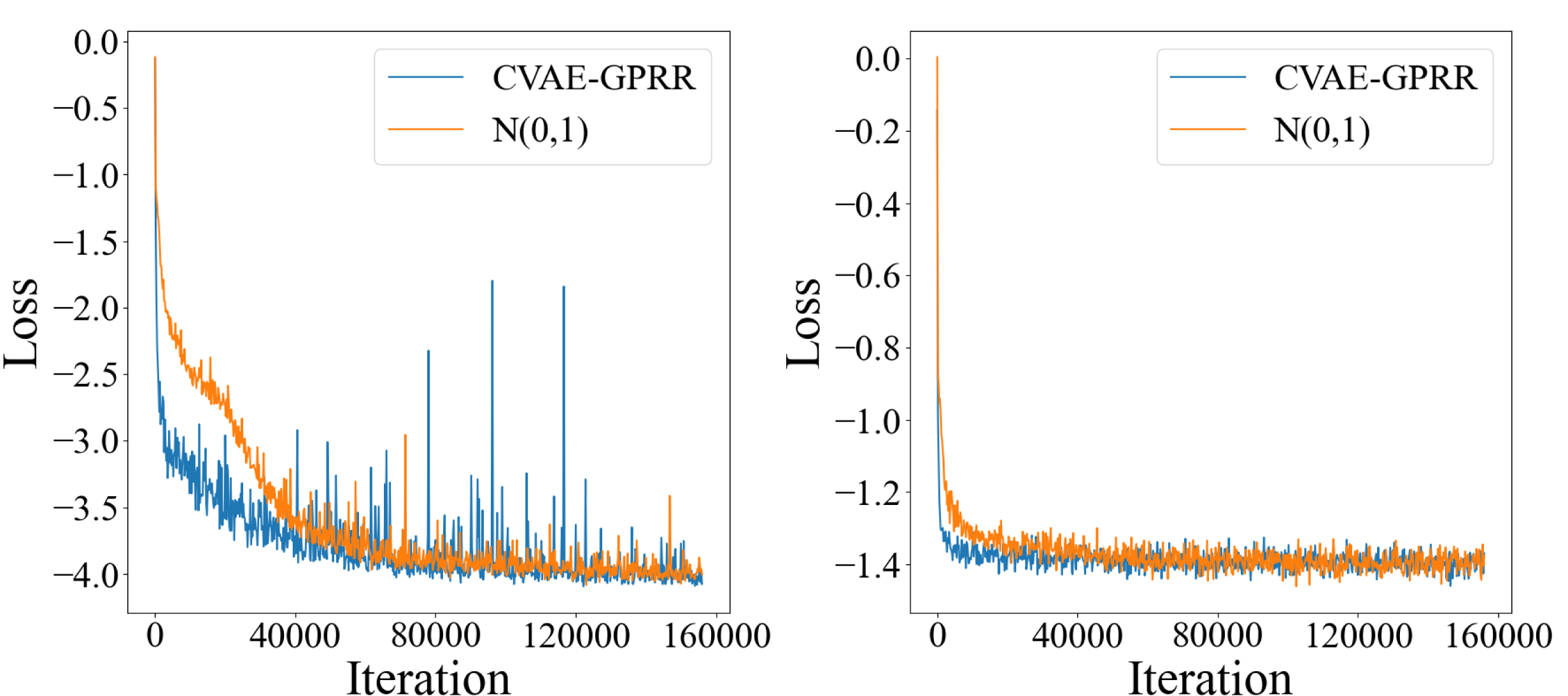}		
	\caption{\small{The parametric diffusion problem case: loss function of likelihood model recorded every 200 iterations. Left: the noise level is 0.01. Right: the noise level is 0.15.}}	
	\label{sec4_exper2:fig_prior_loss}	
\end{figure}
\subsection{Parametric p-Laplacian equation}
\label{ssec:2D parametric p-Laplacian equation}
In this section, we consider a parametric p-Laplacian equation with Dirichlet boundary for a square physical region $\Omega:=(0,1)^2\subset \mathbb{R}^2$:
\begin{equation}
	\left\{
	\begin{aligned}
		&-div\big(|u(\boldsymbol{x},\boldsymbol{\xi})|^{p-2}\nabla u(\boldsymbol{x},\boldsymbol{\xi})\big)=f(\boldsymbol{x},\boldsymbol{\xi})\ &in \ \Omega,\\
		&u(\boldsymbol{x},\boldsymbol{\xi})=g(\boldsymbol{x})\ &on\ \partial{\Omega},
	\end{aligned}
	\right.
	\label{sec4_exper3:equation}
\end{equation}
where $\Vert\cdot\Vert_2$ is 2-norm and the boundary conditions $g$ is defined as the piecewise linear interpolant of the trace of the indicator function $\textbf{1}_{X}(\boldsymbol{x})$. The support of $\textbf{1}_{X}(\boldsymbol{x})$ is the set $X:=\big(\{0\}\times[0.25,0.75]\big)\bigcup\big([0.6,1]\times[0.25,1]\big)$. We introduce the parameter vector
\begin{equation}
	\boldsymbol{\xi}=(p,s,\iota,x_0^1,x_0^2)\in (2,5]\times[5,10]^2\times[-1,1]^2\subset\mathbb{R}^5,
	\nonumber
\end{equation}
and the forcing is
\begin{equation}
	f(\boldsymbol{x},\boldsymbol{\xi})=\frac{\exp(|s|)}{\pi\iota}\exp\bigg(-\frac{\sum_{i=1}^2(x_i-x_0^i)^2}{2\iota^2}\bigg),
	\nonumber
\end{equation}
where $\iota$ is forcing width, $s$ is forcing strength and $(x_0^1,x_0^2)$ is position of forcing. The p-Laplacian operator in the left hand of equation (\ref{sec4_exper3:equation}) is derived from a nonlinear Darcy law and the continuity equation \cite{sec4:p-laplacian}. The parameter p of p-Laplacian is usually greater than or equals to one. For p=2, equation (\ref{sec4_exper3:equation}) is known as Possion equation, which is linear. At critical points ($\nabla u=0$),the equation is degenerate for $p>2$ and singular for $p<2$ (\cite{sec4:fluid}, Chapter1). In this subsection, we only consider former case and let p range from 2 to 5 (Equations with $p>5$ is difficult to be solved).\par
In this experiment, We partition the physical region with a grid of $50\times50$ nodes (fine mesh) and use an efficient solver provided by literature \cite{sec4:equation_solver} to solve the equation for 1500 different parameters, half of which are used as training data and the rest consist of test set. Examples of solutions for three different training parameters are shown in Figure \ref{sec4_exper3:fig_equation_eg}. In addition, a white Gaussian noise with zero means and 0.2 standard deviation is added to the training data and we only take solutions in the coarse mesh of $26\times26$ nodes for training the likelihood model .

\begin{figure}[htbp]		
	\centering		
	\includegraphics[scale=0.25]{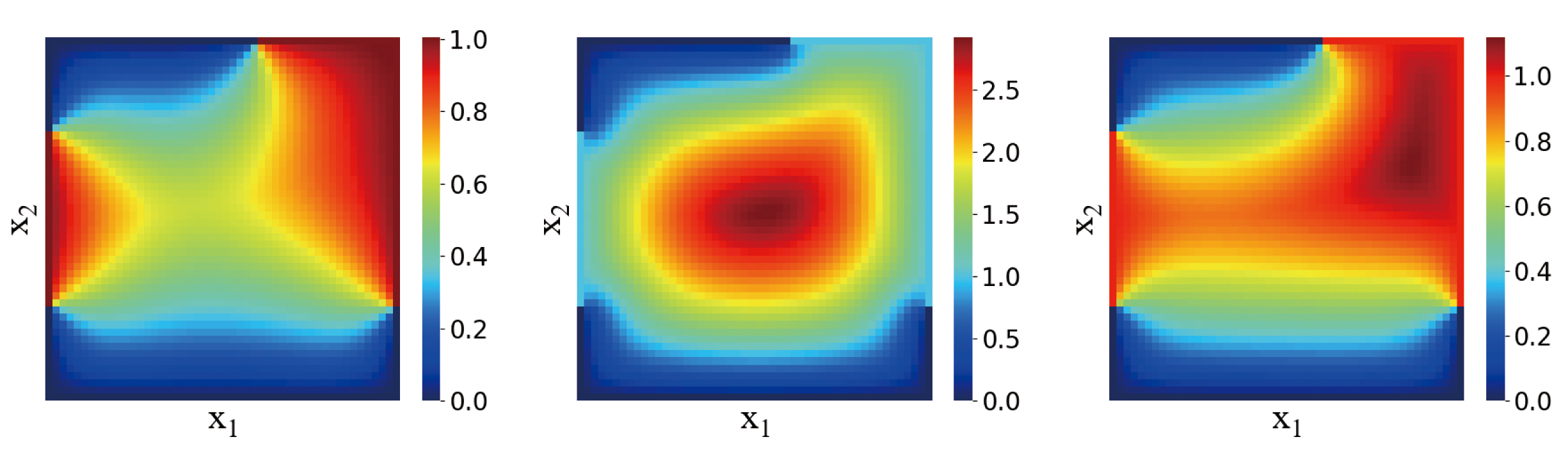}		
	\caption{\small{Examples of solutions (fine mesh) of p-Laplacian equation (\ref{sec4_exper3:equation}) for the parameters  $\boldsymbol{\xi}^{(1)}=(3.19,  5.69 ,  0.91, -0.28,  9.18)$(left), $\boldsymbol{\xi}^{(2)}=(3.10,9.94, -0.92,  0.77,9.57)$}(middle) and  $\boldsymbol{\xi}^{(3)}=(3.23,  6.86, -0.55, -0.11,  6.33)$(right).}	
	\label{sec4_exper3:fig_equation_eg}	
\end{figure}

In this test case, the hyperparameters of likelihood model and training process is shown in Table \ref{sec4_exper3:table_hyper} and ten POD latent variables are chosen.

\begin{table}[htbp]
	\centering  
	\setlength{\abovecaptionskip}{0cm}
	\setlength{\belowcaptionskip}{0.2cm}
	\caption{Hyperparameters of likelihood model for the parametric p-Laplacian equation.}  
	\label{sec4_exper3:table_hyper}  
	\scalebox{0.8}{
		\begin{tabular}{ccccc}
			\toprule  
			Layer structure & Optimizer & Learning rate & Iteration & Batch size \\  
			\midrule
			& & & &  \\[-6pt]  
			17-100-100-100-100-100-100-2& Adam & 0.001-0.0001-0.00001 & 33800-16900-16900 &1000 \\
			\bottomrule
		\end{tabular}
	}
\end{table}
For a given order $p$ and parameters of forcing, the proposed method can predict the mean of the corresponding solution, and also gives uncertainty represented as standard deviation at each spaital location, which is useful when the size of training set is small. In Figure \ref{sec4_exper3:fig_result_datasize}, we show predictions over 1000 POD latent variables samples for one test instance with different size of training data. We can see that accuracy of predictive mean improves as the number of training data increases, while the predictive uncertainty (standard deviation) drops. And only half of original training data can lead to a similar performance in Table \ref{sec4_exper3:table_result_datasize}, which can help reducing the computational cost of GPR.
\begin{figure}[htbp]		
	\centering		
	\includegraphics[scale=0.30]{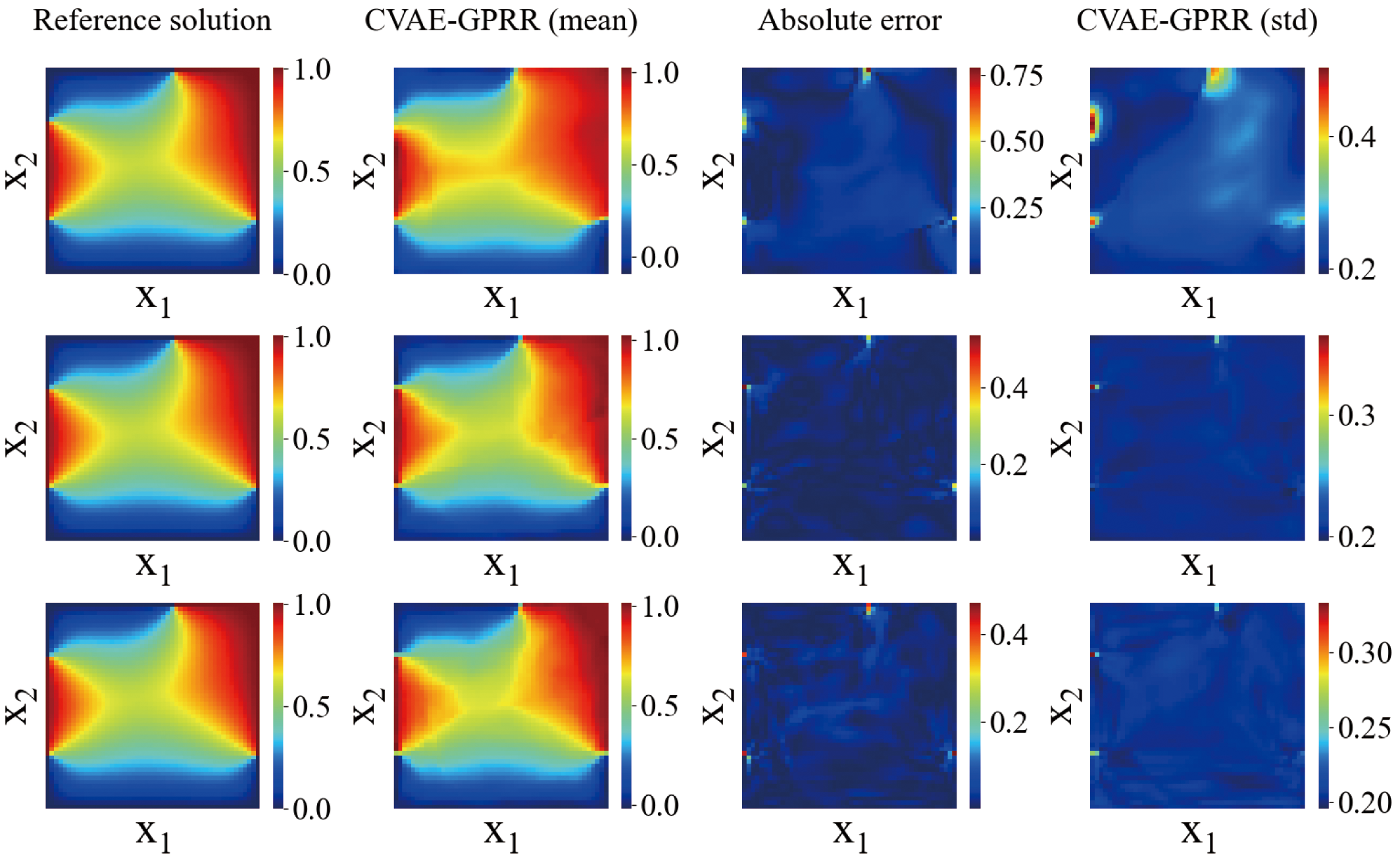}		
	\caption{\small{The parametric p-Laplacian equation case: predictive result (fine mesh) of CVAE-GPRR at the test instance of the parameters $\boldsymbol{\xi}=(3.06,5.50,0.59,-0.07,8.44)$. Top: 250 training data. Middle: 500 training data. Bottom: 1000 training data.}}	
	\label{sec4_exper3:fig_result_datasize}	
\end{figure}
\begin{table}[htbp]
	\centering  
	\setlength{\abovecaptionskip}{0cm}
	\setlength{\belowcaptionskip}{0.2cm}
	\caption{The parametric p-Laplacian equation case: relative test mean errors of CVAE-GPRR with different number of training data.}  
	\label{sec4_exper3:table_result_datasize}  
	\begin{threeparttable}
		\scalebox{0.8}{
			\begin{tabular}{c|cc}
				\toprule  
				\diagbox[width=9cm,height=1.5cm]{the number of training data}{$\epsilon_{test}$}{Method} & \tabincell{c}{CVAE-GPRR\\(coarse mesh)}& \tabincell{c}{CVAE-GPRR\\(fine mesh)}\\   
				\hline
				& &  \\[-6pt]  
				250 &$8.0879\times 10^{-2}$ &$8.1895\times 10^{-2}$ \\
				& & \\[-6pt]  
				500 &$3.9207\times 10^{-2}$ &$4.8534\times 10^{-2}$ \\
				& & \\[-6pt]  
				1000 &$2.9369\times 10^{-2}$ &$4.0701\times 10^{-2}$ \\
				\bottomrule
		\end{tabular}}
	\end{threeparttable}       
\end{table}

\subsection{Skewed lid-driven cavity in parametric region}
\label{ssec:2D skewed lid-driven cavity in parametric region}
Finally, we introduce the steady Navier-Stokes equations, which model the conservation of mass and momentum for a viscous Newtonian incompressible fluid in a parallelogram-shaped cavity $\Omega(\boldsymbol{\xi})$. As shown in Figure \ref{sec4_fig_exper4:physical_region}, the geometry of computational region can be determined by three parameters: sides length $\xi_1\in[1,2]$,$\xi_2\in[1,2]$ and their included angle $\xi_3\in[\frac{\pi}{6},\frac{5\pi}{6}]$. Let the velocity vector and pressure of the fluid are $\vec{v}=\vec{v}(\boldsymbol{x},\boldsymbol{\xi})=(v_x,v_y),p=p(\boldsymbol{x},\boldsymbol{\xi})$, respectively, we consider the parametric equations of the form:
\begin{equation}
	\left\{
	\begin{aligned}
		&\nabla\cdot \vec{v}=0\ &in\  \Omega(\boldsymbol{\xi}),\\
		&-\upsilon(\boldsymbol{\xi})\Delta \vec{v}+\vec{v}\cdot\nabla \vec{v}+\frac{1}{\rho(\boldsymbol{\xi})}\nabla p=\vec{0} \ &in\  \Omega(\boldsymbol{\xi}),\\
		&+b.c.\ &on \ \partial{\Omega(\boldsymbol{\xi})},
	\end{aligned}
	\right.
	\label{sec4_exper4:equation_definition}
\end{equation}
where the boundary conditions are shown in Figure \ref{sec4_fig_exper4:physical_region}. Unit velocity along the horizontal direction is imposed at the top wall of the cavity, and no-slip conditions are enforced on the bottom and the side walls.The pressure is fixed at zero at the lower left corner. Besides, $\upsilon(\boldsymbol{\xi})$ and $\rho(\boldsymbol{\xi})$ represent the dynamic viscosity and uniform density of the fluid. In this experiment, the density $\rho(\boldsymbol{\xi})$ is set unit and the viscosity is computed depending on the geometry through a dimensionless quantity — the Reynold's number ($Re$), which is defined as $Re=max\{\xi_1,\xi_2\}/\upsilon$. Here, we take $Re=400$.

\begin{figure}[htbp]
	\centering
	{\includegraphics[scale=0.13]{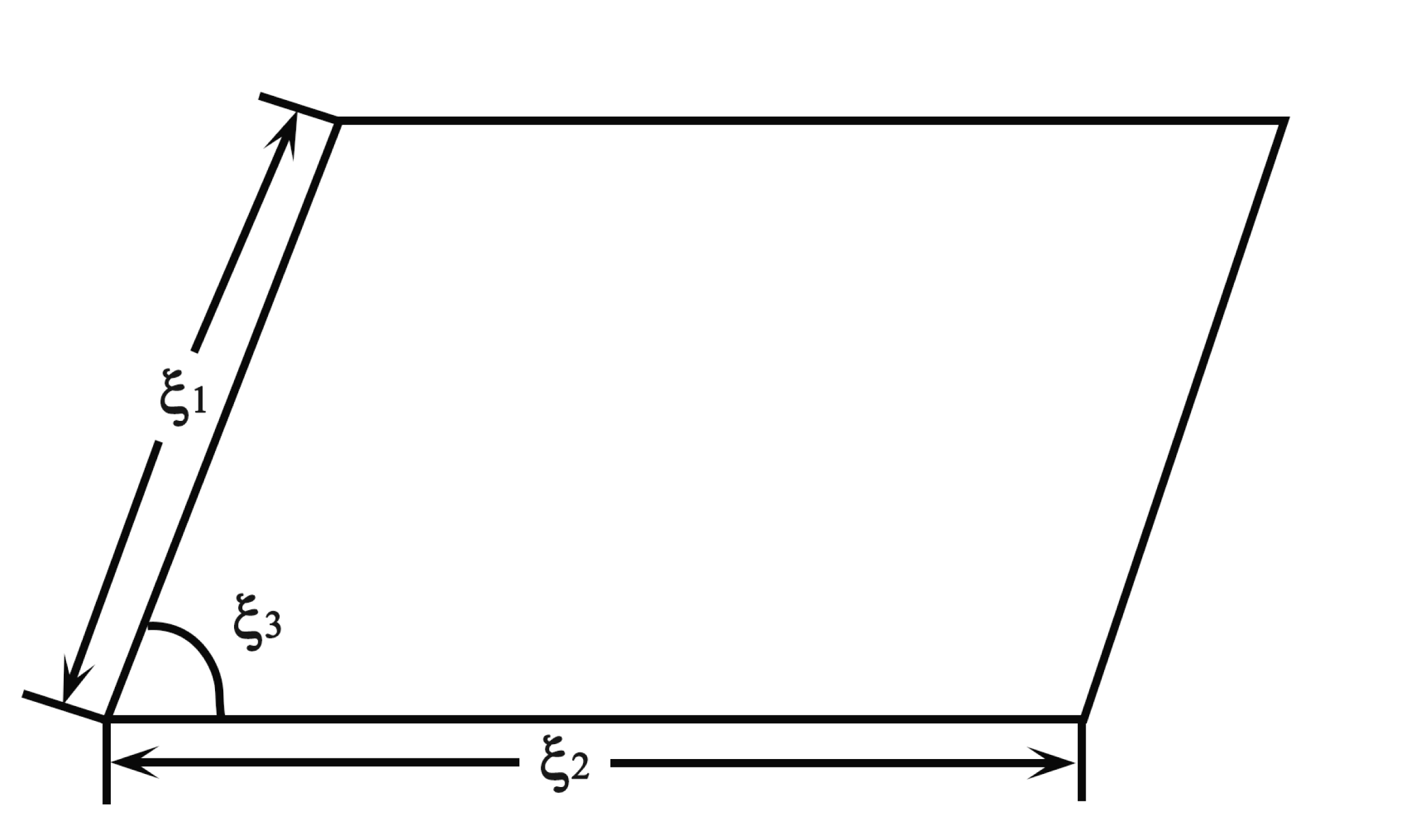}}
	\quad
	{\includegraphics[scale=0.13]{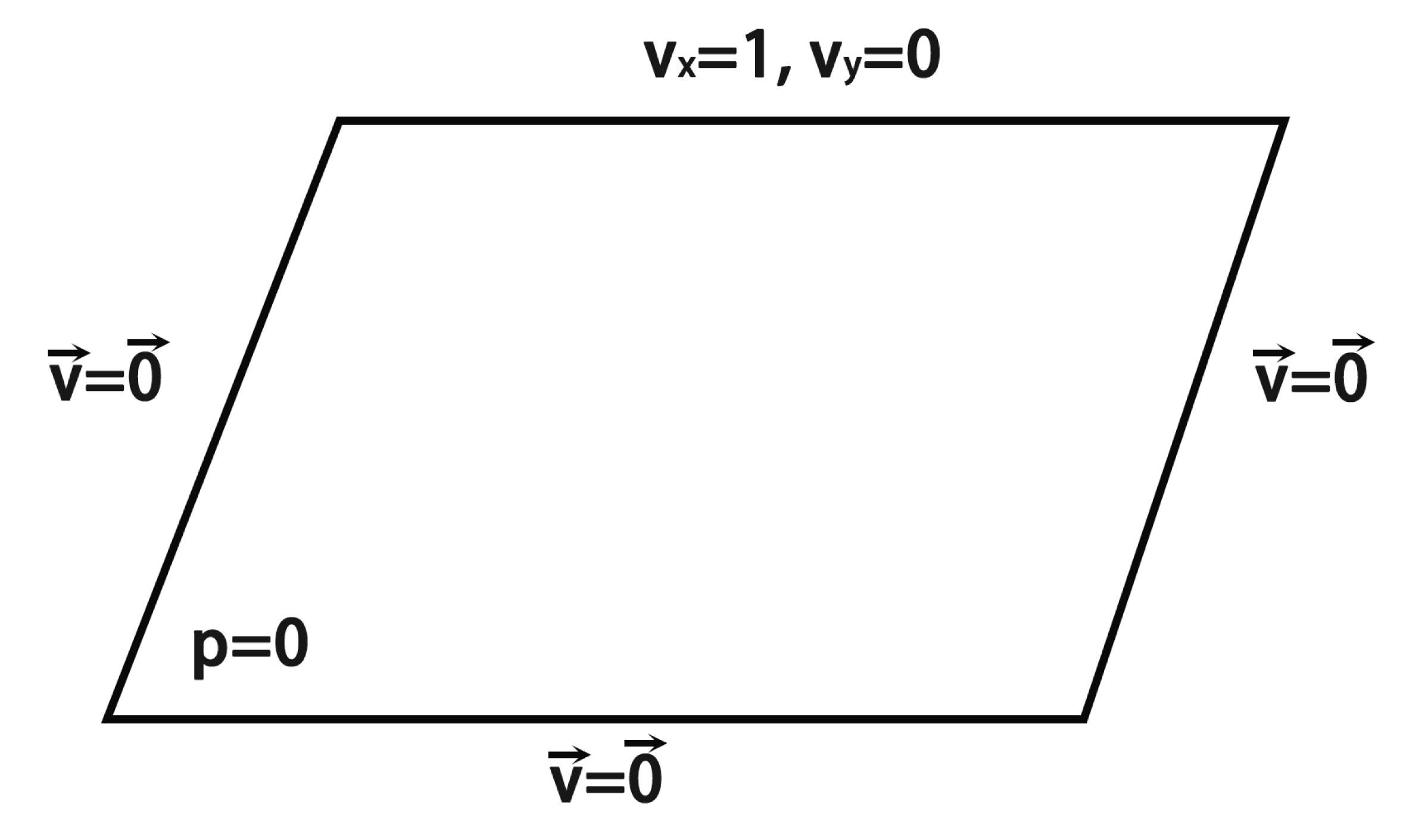}}
	\caption{Parametric geometry of physical region $\Omega$ (left). Boundary condition in equation (\ref{sec4_exper4:equation_definition})(right).}
	\label{sec4_fig_exper4:physical_region}
\end{figure}

We use the dataset provided by \cite{sec4:skewed_lid_driven_cavity,sec4:skewed_lid_driven_cavity_origin}, which simulate equations (\ref{sec4_exper4:equation_definition}) in a mesh with $100\times100$ nodes for 200 different parameters. We choose 150 samples as training set ,the rests are used to test the generalization ability with respect to $\boldsymbol{\xi}$. Although the states of equation includes both pressure and velocity vector, we only use the data of velocity vector in this experiment and model for each entry of $\vec{v}$ respectively in a coarser mesh with $51\times51$ nodes. An example of velocity vector for one training parameters are shown in Figure \ref{sec4_exper4:example}.

\begin{figure}[htbp]
	\centering
	\subcaptionbox{Velocity along x-axis\label{sec4_subfig:v_x}}
	{\includegraphics[scale=0.17]{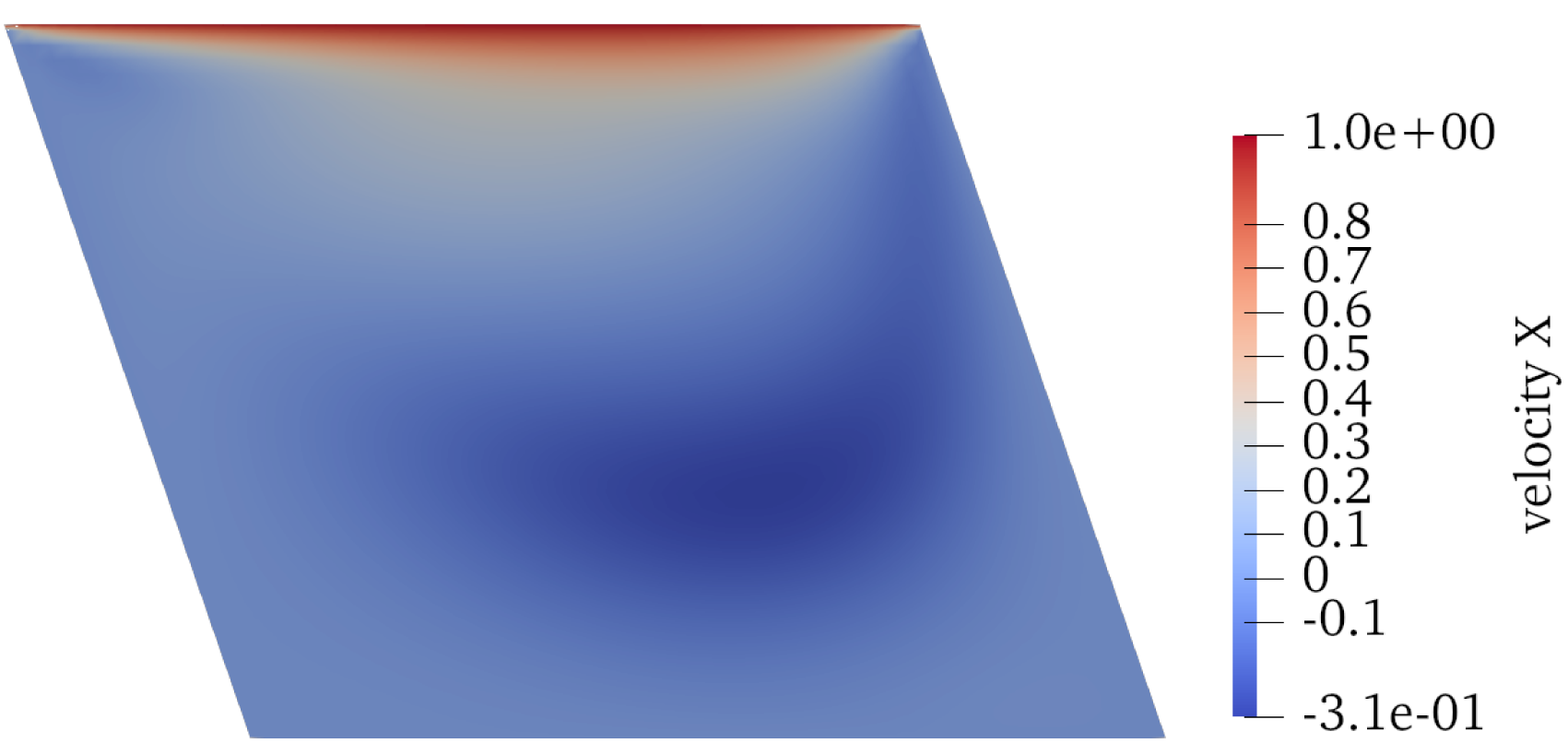}}
	\quad
	\subcaptionbox{Velocity along y-axis\label{sec4_subfig:v_y}}
	{\includegraphics[scale=0.17]{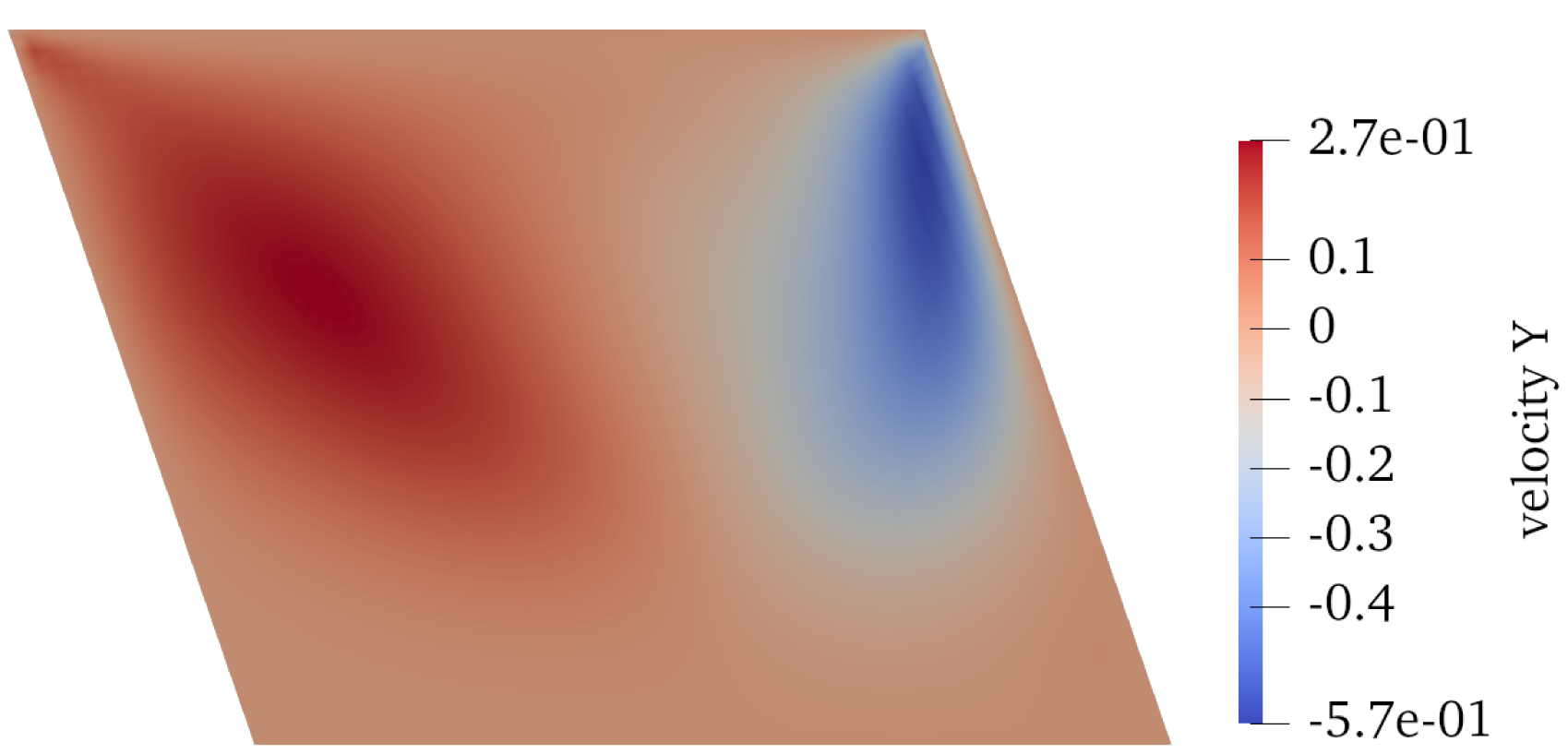}}
	\subcaptionbox{Velocity magnitude\label{sec4_subfig:v_magnitude}}
	{\includegraphics[scale=0.17]{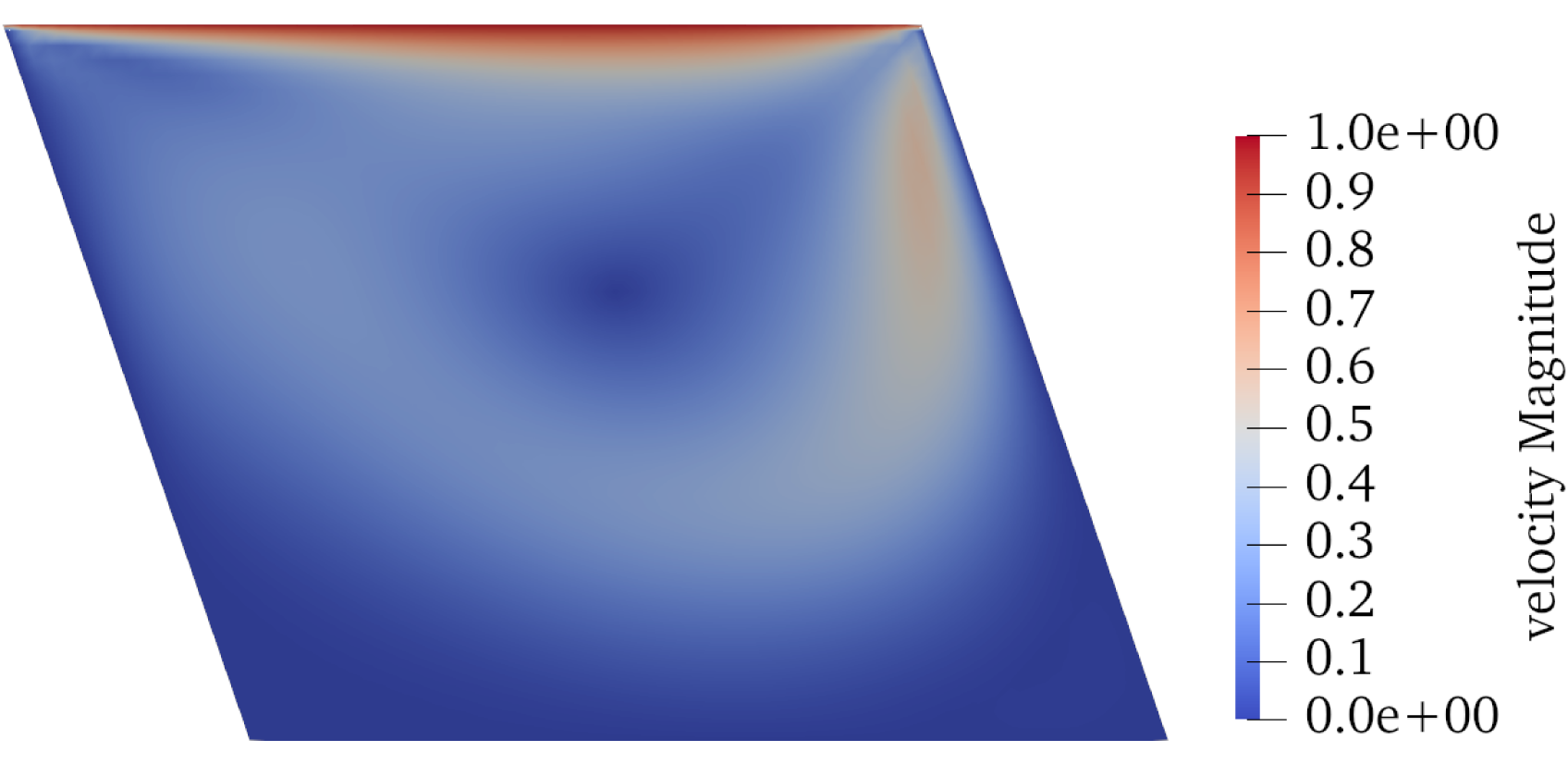}}
	\quad
	\subcaptionbox{Streamlines\label{sec4_subfig:v_streamline}}
	{\includegraphics[scale=0.17]{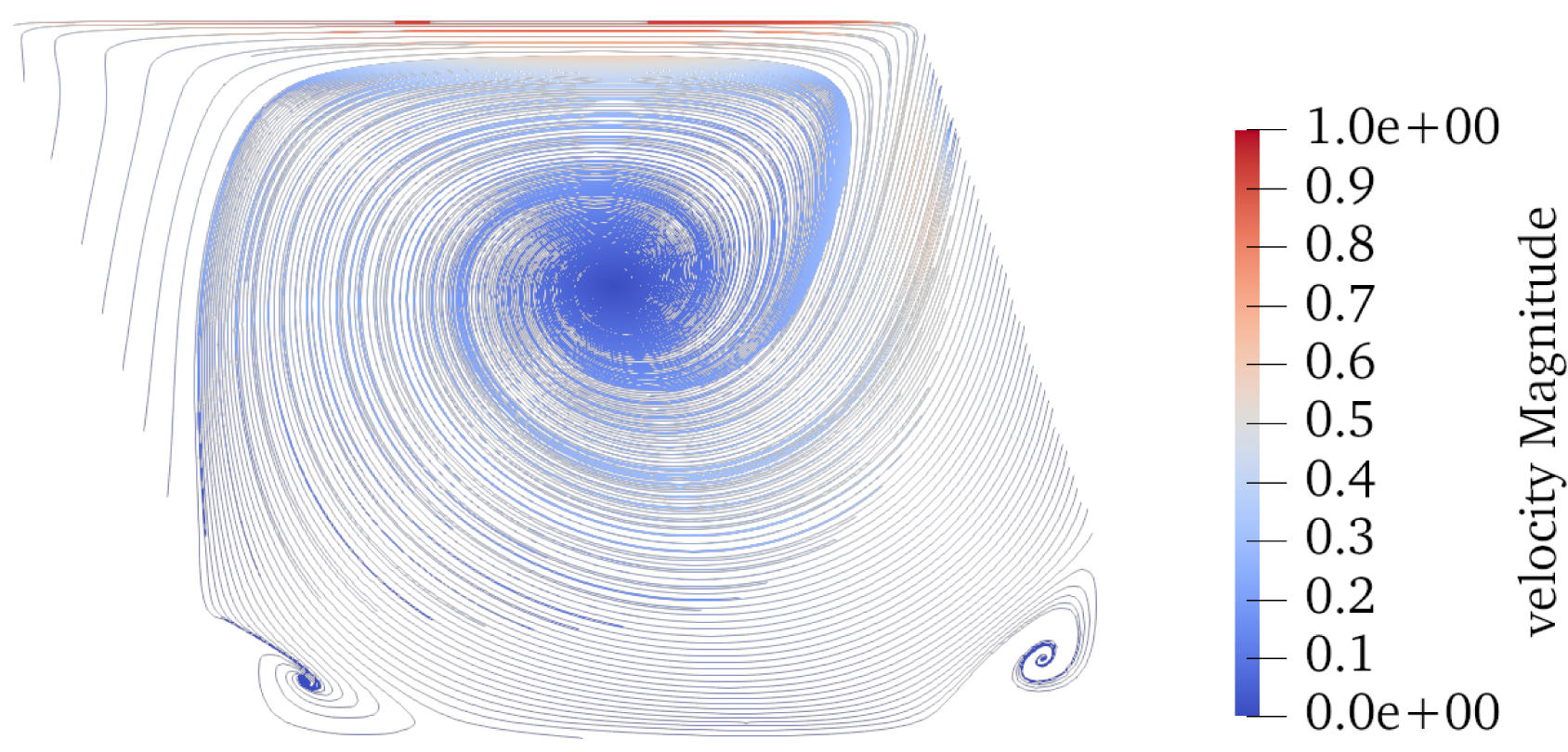}}
	\caption{An example of velocity vector for $\boldsymbol{\xi}=(1.95,1.97,1.40)$.}
	\label{sec4_exper4:example}
\end{figure}

In this test case, the hyperparameters of likelihood model and training process is shown in Table \ref{sec4_exper4:table_hyper} and twenty POD latent variables are chosen for both $v_x$ and $v_y$.

\begin{table}[ht]
	\centering  
	\setlength{\abovecaptionskip}{0cm}
	\setlength{\belowcaptionskip}{0.2cm}
	\caption{Hyperparameters of likelihood model for the skewed lid-driven cavity in parametric region.}  
	\label{sec4_exper4:table_hyper}  
	\scalebox{0.8}{
		\begin{tabular}{ccccc}
			\toprule  
			Layer structure & Optimizer & Learning rate & Iteration & Batch size \\  
			\midrule
			& & & &  \\[-6pt]  
			25-100-100-100-100-2 & Adam & 0.001-0.0001-0.00001 & 80000-80000-80000 &1000 \\
			\bottomrule
		\end{tabular}
	}
\end{table}
We add the Gaussian noise with zero means and 0.1 standard deviation to the data. Figure \ref{sec4_exper4:v_x_0_1} and Figure \ref{sec4_exper4:v_y_0_1} show the mean and standard deviations of velocity in x-axis and y-axis, respectively. The results illustrated that CVAE-GPRR can obtain the distribution whose mean and standard deviation can approximate the distribution of observation data.Thus, the trained likelihood model can be used to generate samples that can simulate the observation data. The approximation error of mean can be further described by relative test mean error $\epsilon_{test}$, which is shown in Table \ref{sec4_exper4:table_relative_test_mean_error}.

\begin{figure}[htbp]
	\centering
	\includegraphics[scale=0.48]{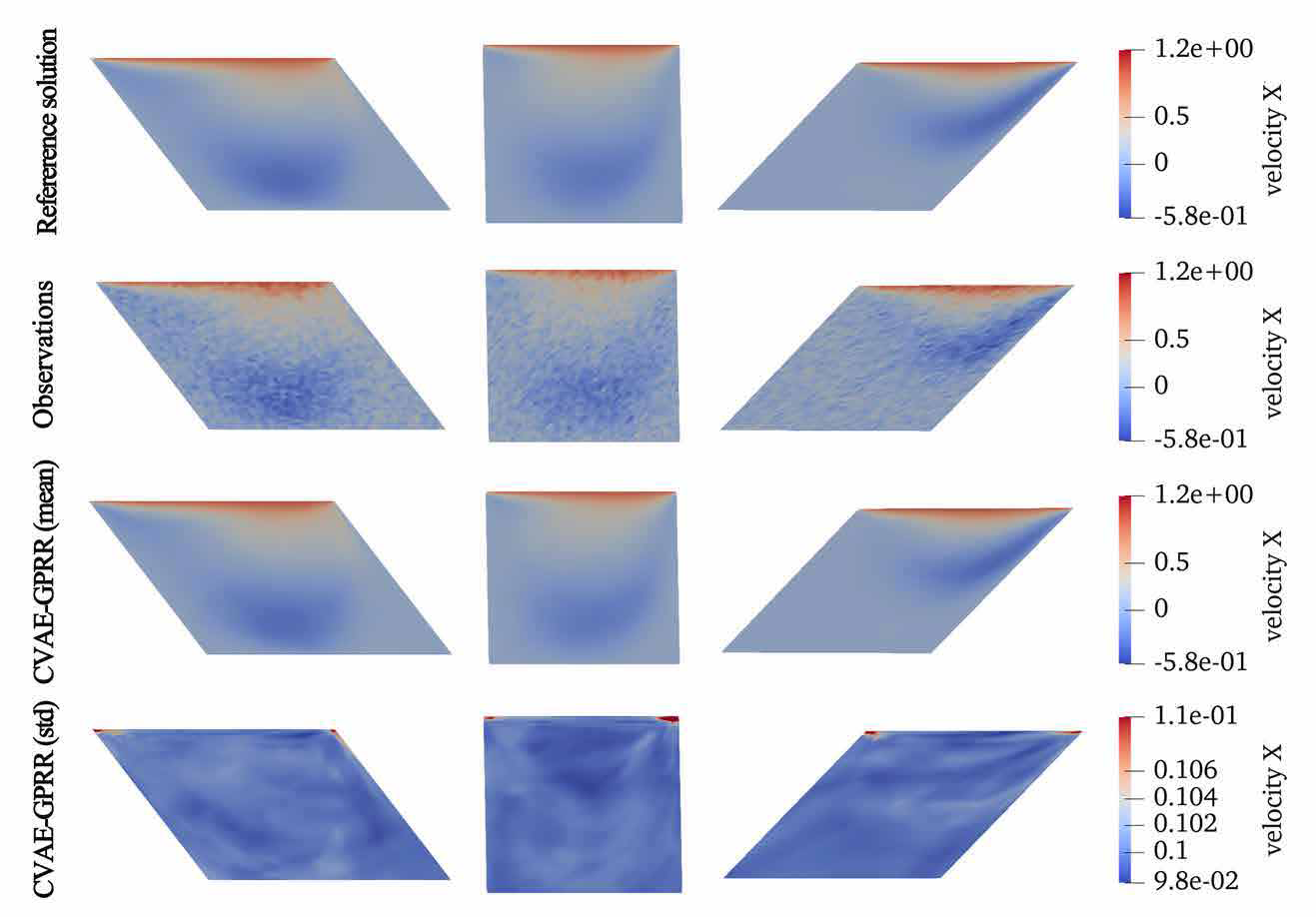}\\
	\caption{The skewed lid-driven cavity in parametric region: predictive results of $v_x$ at the test instances of the parameters (noise level: 0.1). Left: $\boldsymbol{\xi}=(1.646,1.29,2.23)$. Middle: $\boldsymbol{\xi}=(1.37,1.24, 1.59)$. Right: $\boldsymbol{\xi}=(1.92,1.82,0.80)$.}
	\label{sec4_exper4:v_x_0_1}
\end{figure}
\begin{figure}[htbp]
	\centering
	\includegraphics[scale=0.48]{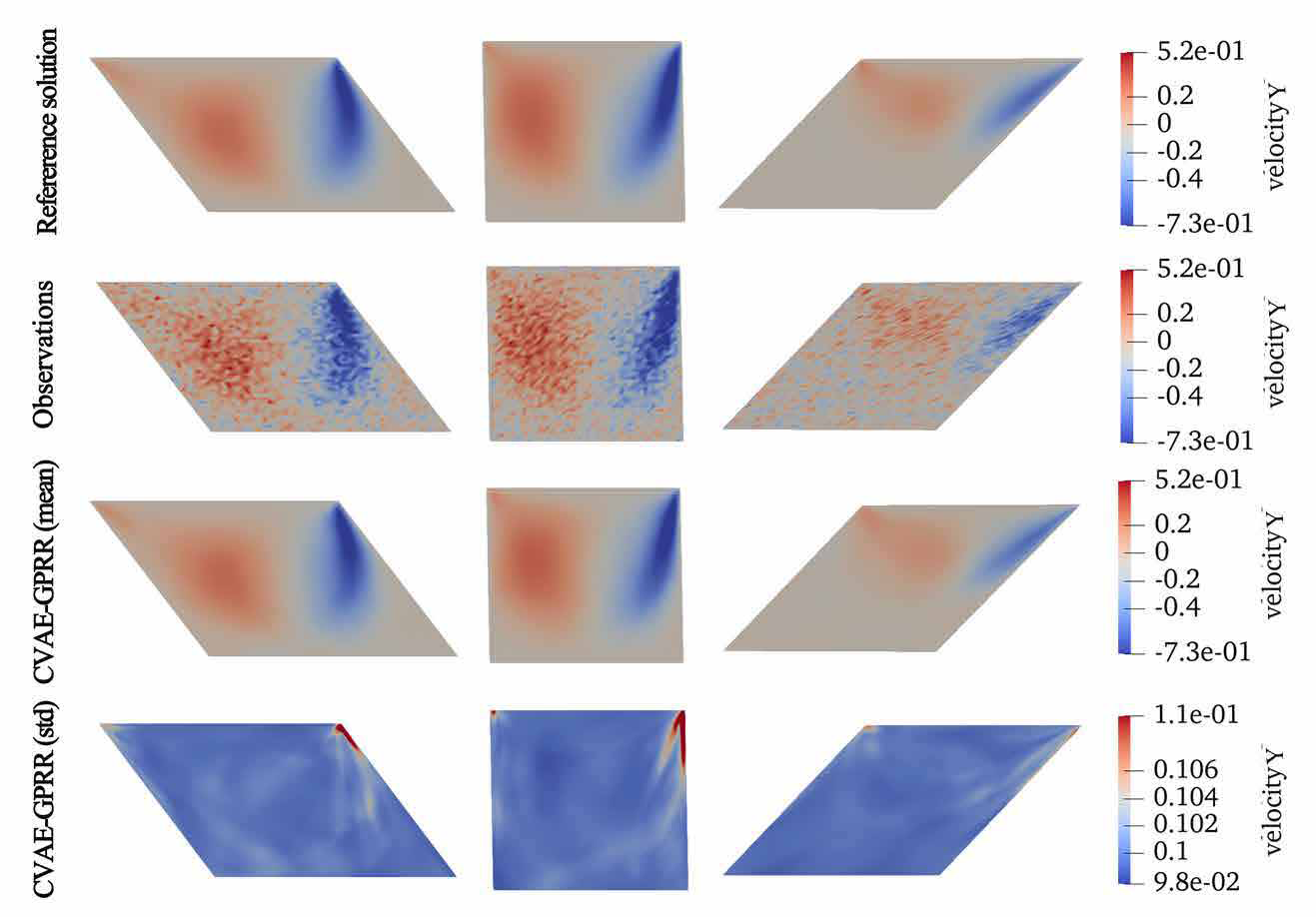}\\
	\caption{The skewed lid-driven cavity in parametric region: predictive results of $v_y$ at the test instances of the parameters (noise level: 0.1). Left: $\boldsymbol{\xi}=(1.646,1.29,2.23)$. Middle: $\boldsymbol{\xi}=(1.37,1.24, 1.59)$. Right: $\boldsymbol{\xi}=(1.92,1.82,0.80)$.}
	\label{sec4_exper4:v_y_0_1}
\end{figure}
\begin{table}[htbp]
	\centering  
	\setlength{\abovecaptionskip}{0cm}
	\setlength{\belowcaptionskip}{0.2cm}
	\caption{The skewed lid-driven cavity in parametric region: relative test mean errors for different noise levels}
	\label{sec4_exper4:table_relative_test_mean_error}  
	\scalebox{0.8}{
	\begin{tabular}{c|cc}
		\toprule  
		\diagbox[width=5cm,height=1.5cm]{Velocity}{$\epsilon_{test}$}{Noise level} & 0.02  & 0.1\\    
		\hline
		& &  \\[-6pt]  
		$v_x$ & $4.1473\times 10^{-2}$ & $5.5369\times 10^{-2}$ \\
		& &  \\[-6pt]  
	    $v_y$ & $5.4056\times 10^{-2}$ & $9.9141\times 10^{-2}$\\
		\bottomrule
    \end{tabular}}
\end{table}
Figure \ref{sec4_exper4:v_x_0_1} and Figure \ref{sec4_exper4:v_y_0_1} hardly visualize the minor difference between the truth and predictive mean. Therefore, the streamlines are further reported in Figure \ref{sec4_exper4:streamline_0_1}. We find that the proposed method can uncover the main circulation zone (regions of high-velocity) from noisy data while fails to approximate the recirculation zones (regions of low-velocity). We attribute the poor performance in the recirculation zones to the influence of noise since the streamlines of observation data shown in the second row of Figure \ref{sec4_exper4:streamline_0_1} only provide the screwy shape of the main circulation zone.
\begin{figure}[htbp]
	\centering
	{\includegraphics[scale=0.48]{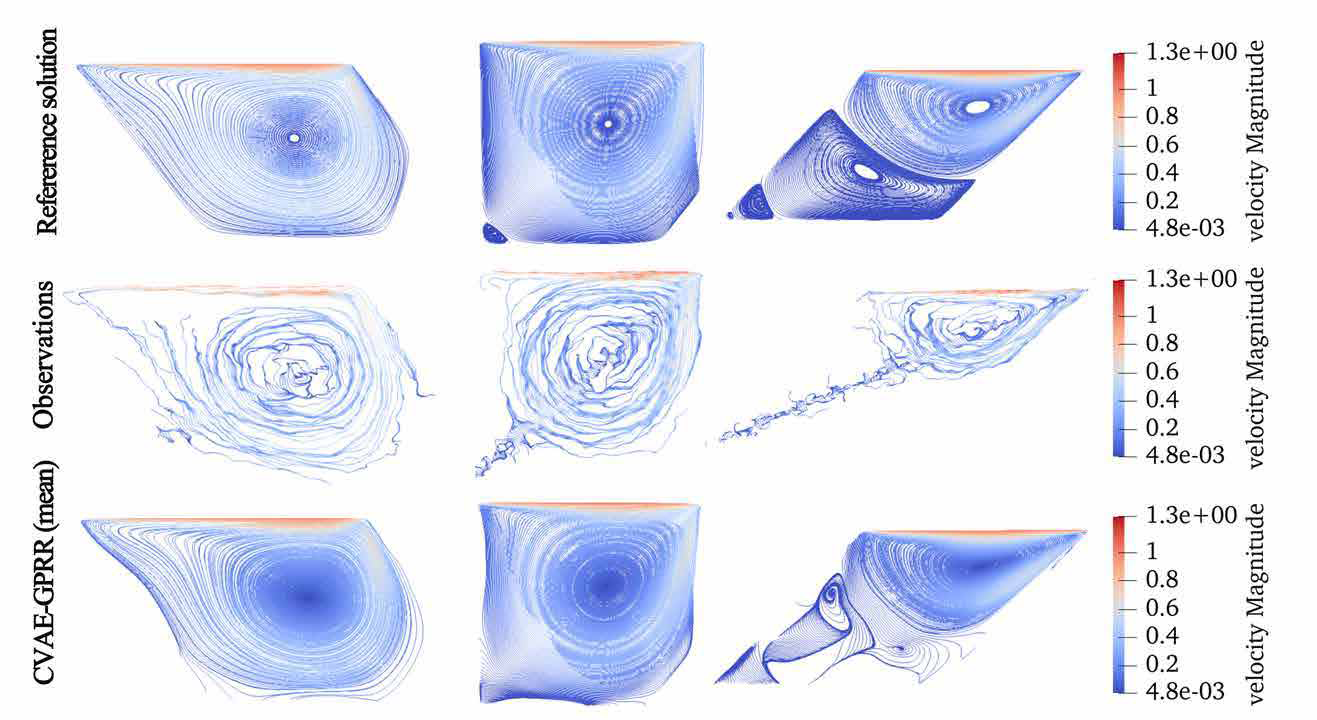}}
	\caption{The skewed lid-driven cavity in parametric region: streamlines colored by velocity magnitude at the test instances of the parameters (noise level: 0.1). Left: $\boldsymbol{\xi}=(1.64,1.29,2.23)$. Middle: $\boldsymbol{\xi}=(1.37,1.24, 1.59)$. Right: $\boldsymbol{\xi}=(1.92,1.82,0.80)$.}
	\label{sec4_exper4:streamline_0_1}
\end{figure}

\section{Conclusion}
In this work, we have provided a conditional variational autoencoder with GPR recognition model. By combing modern deep latent variable modeling with traditional data-driven ROM techniques , CVAE-GPRR inherited their merits and also overcame some limitations, which led to an efficient modeling for parametric problems. For recognition model, CVAE uses a neural network to extract low-dimensional features and alleviate the influence of noise in data, which is not interpretable and is hard to train due to a lot of parameters in neural networks. Instead, in the recognition process, the proposed method first used POD to obtain the principle modes. POD projection coefficients were chosen as latent variables. And then non-parametric probabilistic model GPR was used to learn the mapping from parameters space to latent space and denoised the observation data at the same time. Since each entry of POD projection coefficients were uncorrelated, we could train the GPR recognition model parallelly. Compared to traditional non-intrusive ROM techniques, we replaced the POD basis with likelihood neural networks to reconstruct data when data was highly noisy. To obtain a model that also had generalization abilities of physical variables, which was unavailable in both CVAE and non-intrusive ROM, we added the physical variables into the inputs of likelihood model. With efficacy of CVAE-GPRR numerically validated by various examples, CVAE-GPRR was shown to be a powerful tool for modeling parametric systems with noisy observations and could give the uncertainty quantification at the same time.

\section{Aknowledge}
L. Jiang acknowledges the support of NSFC 12271408 and the
Fundamental Research Funds for the Central Universities.

\smallskip
\bigskip
\textbf{Conflict of interest statement: }
The authors have no conflicts of interest to declare. All co-authors have seen and agree with the contents of the manuscript.


\end{document}